\documentclass[12pt, a4paper]{article}
\usepackage{amsthm,amsfonts,amsmath,amssymb,bookmark,appendix}

\theoremstyle{plain}
\newtheorem{theorem}{Theorem}
\newtheorem{lemma}[theorem]{Lemma}

\theoremstyle{definition}
\newtheorem{definition}[theorem]{Definition}

\theoremstyle{remark}
\newtheorem*{remark}{Remark}

\DeclareMathOperator{\spec}{spec}
\DeclareMathOperator{\ess}{ess}
\DeclareMathOperator{\supp}{supp}
\DeclareMathOperator{\Tr}{Tr}

\DeclareMathOperator{\loc}{loc}

\def\geqslant{\ge}
\def\leqslant{\le}
\def\bq{\begin{eqnarray}}
\def\eq{\end{eqnarray}}
\def\bqq{\begin{eqnarray*}}
\def\eqq{\end{eqnarray*}}
\def\nn{\nonumber}
\def\minus {\backslash}
\def\eps{\varepsilon}

\newcommand{\norm}[1]{\lVert #1 \rVert}
\newcommand{\clmb}{|\,.\,|}

\def\TF{\mathrm{TF}}

\def\R {\mathbb{R}}

\def\E {\mathcal{E}}

\title{\bf Asymptotics for Two-dimensional Atoms}
\author{Phan Thanh Nam, Fabian Portmann and Jan Philip Solovej}

\begin{document}
\maketitle

\begin{abstract} 
We prove that the ground state energy of an atom confined to two dimensions with an infinitely heavy nucleus of charge $Z>0$ and $N$ quantum electrons of charge $-1$ is $E(N,Z)=-\frac{1}{2}Z^2\ln Z+(E^{\TF}(\lambda)+\frac{1}{2}c^{\rm H})Z^2+o(Z^2)$ when $Z\to \infty$ and $N/Z\to \lambda$, where $E^{\TF}(\lambda)$ is given by a Thomas-Fermi type variational problem and $c^{\rm H}\approx -2.2339$ is an explicit constant. We also show that the radius of a two-dimensional neutral atom is unbounded when $Z\to \infty$, which is contrary to the expected behavior of three-dimensional atoms.
\vspace{5pt}


\end{abstract}

\tableofcontents
\newpage

\section{Introduction}
We consider an atom confined to two dimensions. It has a fixed nucleus of charge $Z>0$ and $N$ non-relativistic quantum electrons of charge $-1$. For simplicity we shall assume that electrons are spinless because the spin only complicates the notation and our coefficients in an obvious way. The system is described by the Hamiltonian
\[
H_{N,Z} = \sum_{i = 1}^N {\left( { - \frac{1}{2}\Delta _i  - \frac{Z}
{{|x_i |}}} \right)}  + \sum_{1 \le i < j \le N} {\frac{1}
{{|x_i  - x_j |}}} 
\]
acting on the antisymmetric space $\bigwedge_{i = 1}^N L^2 (\mathbb{R}^2)$. Note that we are using the three-dimensional Coulomb potential to describe the confined atom. The ground state energy of the system is the bottom of the spectrum of $H_{N,Z}$, denoted by
\[
E(N,Z) = \inf {\text{spec}}~ H_{N,Z}=\inf_{\norm{\psi}_{L^2}  = 1} (\psi ,H_{N,Z} \psi ).
\]

One possible approach to obtain the above Hamiltonian is to consider a three-dimensional atom confined to a thin layer $\R^2 \times (-a,a)$ in the limit $a \to 0^+$ (see \cite{D10}, Section~3, for a detailed discussion on the hydrogen case).

To the best of our knowledge, there is no existing result on the ground state energy and the ground states of the system, except for the case of hydrogen \cite{Y91,PP02}. The purpose of this article is to give a rigorous analysis for large $Z$-atom asymptotics and our main results are the following theorems.

\begin{theorem}[Ground state energy] \label{thm:GSE} Fix $\lambda>0$. When $Z\to \infty$ and $N/Z\to \lambda$ one has  
$$E(N,Z)=-\frac{1}{2}Z^2\ln Z+ \left( {E^{\TF}(\lambda)+\frac{1}{2} c^{\rm H}} \right) Z^2+o(Z^2)$$
where $E^{\TF}(\lambda)$ is the Thomas-Fermi energy (defined in Section~\ref{sec:TF-theory}) and $c^{\mathrm{H}}  =  - 3\ln (2) - 2\gamma _E  + 1 \approx -2.2339$ with $\gamma_{E}  \approx 0.5772$ being Euler's constant \cite{Eu1736}. In particular, $\lambda \mapsto E^{\TF}(\lambda)$ is strictly convex and decreasing on $(0,1]$ and $E^{\TF}(\lambda)=E^{\TF}(1)$ if $\lambda\ge 1$.
\end{theorem}

\begin{remark}

By considering the hydrogen semiclassics we conjecture that the next term of $E(\lambda Z,Z)$ is of order $Z^{3/2}$. In contrast, the ground state energy in three dimensions behaves as
$$E(Z,Z) = c^{\TF}Z^{7/3} + c^{\mathrm{S}}Z^2 + c^{\mathrm{DS}}Z^{5/3} + o(Z^{5/3}),$$
where the leading (Thomas-Fermi \cite{Thomas,Fermi}) term was established in \cite{LS77}, the second (Scott \cite{Scott}) term was proved in \cite{Hu90,SW87}, and the  third (Dirac-Schwinger \cite{Dirac,Schwinger}) term was shown in \cite{FS90}.

\end{remark}

For three-dimensional atoms the leading term in the energy asymptotics of 
order $Z^{7/3}$ may be understood entirely from semiclassics. The 
contribution to this term comes from the bulk of the electrons located 
mainly at a distance of order $Z^{-1/3}$ from the nucleus. The term of 
order $Z^2$, the Scott term, is a pure quantum correction coming from the 
essentially finitely many inner most electrons at a distance of order 
$Z^{-1}$ from the nucleus. 

In the two-dimensional case the situation is more 
complicated. The leading term of order $Z^2\ln(Z)$ is semiclassical and 
comes from the fact that the semiclassical integral is logarithmically 
divergent, but has a natural cut-off at a distance of order $Z^{-1}$ from 
the nucleus. The term of order $Z^2$ has two contributions. One part is 
semiclassical and comes essentially from electrons at distances of order 
$1$ from the nucleus and another part, which corresponds to the three-dimensional Scott 
correction, coming from the essentially finitely many inner most electrons 
at a distance of order $Z^{-1}$ from the nucleus.

Thus the two-dimensional atom has two regions. The innermost region of size 
$Z^{-1}$ contains a finite number of electrons and contributes with 
$Z^2$ to the total energy. The outer region from $Z^{-1}$ to order $1$ has 
a high density of electrons and can be understood semiclassically. It 
contributes to the energy with both $Z^2\ln(Z)$ from the short distance 
divergence and with $Z^2$ from the bulk at distance $1$.

\begin{theorem}[Extensivity of neutral atoms]\label{thm:radius-atom} Assume that $N/Z\to 1$ and $\Psi_{N,Z}$ is a ground state of $H_{N,Z}$. Then, for any $R>0$ there exists $C_R>0$ such that 
\[
 \int_{|x| \ge R} {\rho _{\Psi _{N,Z} } (x)dx} \ge C_RZ+o(Z). 
\]
\end{theorem}

\begin{remark}  If we define the radius $R_{Z}$ of a neutral atom ($N=Z$) by 
\[
\int_{ |x| \ge R_{Z}} {\rho _{\Psi_{Z,Z}}(x)dx}  = 1
\]
then Theorem~\ref{thm:radius-atom} implies that $\lim_{Z\to \infty}R_{Z}= \infty$. In three dimensions, however, the radius is expected to be bounded independently of $Z$ (see \cite{SSS90, So03}). 
\end{remark}

Our main tool to understand the ground state energy and the ground states is the Thomas-Fermi (TF) theory introduced in Section~\ref{sec:TF-theory}. In this theory, the $Z$-ground state scales as $Z \rho^{\TF}_{N/Z}(x)$ and the absolute ground state $\rho^{\TF}_1$ (when $N=Z$) has unbounded support. Roughly speaking, the extensivity of the TF ground state implies the extensivity of neutral atoms (in contrast, the three-dimensional TF $Z$-ground state scales as $Z^{2} \rho^{\TF}(Z^{1/3}x)$, i.e. its core \textit{shrinks} as $Z^{-1/3}$). 

The challenging point of the two-dimensional TF theory is that the TF potential $V^{\TF}(x)$ is {\it not} in $L^2_{loc}(\R^2)$ (it behaves like $|x|^{-1}$ near the origin).  Consequently, one {\it cannot} write the semiclassics  of  $\Tr \left[ { -h^2\Delta  - V^{\TF}} \right]_ -$ in the usual way because
\[
(2\pi )^{ - 2} \iint {[h^2 p^2  - V^{\TF } (x)]_ -  dpdx} = - (8\pi h^2)^{-1}\int {[V^{\TF } (x)]_ + ^2 dx}=-\infty.
\]
This property complicates matters in the semiclassical approximation. In contrast, the three-dimensional semiclassical approximation leads to the behavior $-(15\pi ^2 h^3)^{-1}\int_{\R^3} [V^{\TF } (x)]_+^{5/2}dx$ which is finite for the Coulomb singularity $V^{\TF}(x)\sim |x|^{-1} \in L^{5/2}_{\loc}(\R^3)$. 

We shall follow the strategy of proving the Scott's correction given by Solovej and Spitzer \cite{SS03} (see also \cite{SSS10}), that is to compare the semiclassics of TF-type potentials with hydrogen. More precisely, in the region close to the origin we shall compare directly with hydrogen, whereas in the exterior region we can employ the coherent state approach. We do not use the new coherent state approach introduced in \cite{SS03}, since the usual one \cite{Li81,Th81} is sufficient for our calculations. In fact, we prove the following semiclassical estimate for potentials with Coulomb singularities.

\begin{theorem}[Semiclassics for Coulomb singular potentials]\label{thm:semiclassics-general-potential} 
Let $V\in L^2_{\loc}(\R^2\minus \{0\})$ be a real-valued potential such that $1_{\{|x|\ge 1 \}}V_+\in L^2(\R^2)$ and 
$$ |V(x)-\kappa |x|^{-1}|\le C |x|^{-\theta}~~\text{for all}~|x|\le \delta,$$
where $\kappa>0$, $\delta>0$, $1>\theta>0$ and $C>0$ are universal constants.  Then, as $h\to 0^+$,
\bqq
	\Tr \left[ { -h^2\Delta  - V} \right]_ -   &=& -(8\pi h^2)^{-1}\int_{\mathbb{R}^2 } 
	{\left( {[V(x)]_ + ^2  - \kappa^2 [|x|^{ - 1}  - 1 ]_ + ^2 } \right)dx}\\
	&~&+ \kappa^2 (4h^2)^{ -1} \left[ {\ln(2\kappa^{-1}h^2) +c^{\mathrm{H}}} \right] + o(h^{ - 2} ),
\eqq
where $c^{\mathrm{H}}  =  - 3\ln (2) - 2\gamma _E  + 1 \approx -2.2339$ with $\gamma_{E}  \approx 0.5772$ being Euler's constant \cite{Eu1736}.
\end{theorem}

The article is organized as follows. In Section~\ref{sec:summary} we give a brief summary of the existing results concerning atoms confined to two dimensions. Section~\ref{sec:TF-theory} contains basic information on the TF theory. The most technical part of the article is in Section~\ref{sec:Semiclassics-TF-potential}, where we show the semiclassics for the TF potential. The main theorems are proved in Section~\ref{sec:results}. Some technical proofs are deferred to the Appendix.

\section{Preliminaries}\label{sec:summary}
\subsection{Spectral Properties}\label{sec:spectral}

For completeness, we start by collecting some basic properties of the spectrum of $H_{N,Z}$, whose proofs can essentially be adapted from the usual three-dimensional case (see the Appendix).  

\begin{theorem}[Spectrum of $H_{N,Z}$]\label{thm:HVZ} 
Let $H_{N,Z}$ be the operator defined above.
\begin{itemize}
	\item[$(i)$] (HVZ Theorem) The essential spectrum of $H_{N,Z}$ is
	$$\ess\spec H_{N,Z}=[E(N - 1,Z),\infty).$$
	Consequently, for non-vanishing binding energy $E(N-1)-E(N)=:\eps>0$, the operator $H_{N,Z}$ has (at least) one ground state. Moreover, in 
	this case any ground state $\Psi_{N,Z}$ of $H_{N,Z}$ has exponential decay as
	\[
	\rho _{\Psi _{N,Z} } (x) \le C|x|^{\frac{{4Z}}
	{{\sqrt {2\eps } }} -2} e^{ - 2\sqrt {2\eps } |x|} ~~\text{for $|x|$ large}
	\]
	where the density $\rho _{\Psi _{N,Z} }$ is defined as in Section \ref{subsec:Useful-Inequalities}.
	
	\item[$(ii)$] (Zhislin's Theorem) If $N<Z+1$ then the binding condition $E(N)<E(N-1)$ is satisfied, and hence $H_{N,Z}$ has  a ground state. 
	
	\item[$(iii)$] (Asymptotic neutrality)  The largest number $N=N_c(Z)$ of electrons such that $H_{N,Z}$ has a ground state is finite and satisfies $
	\lim_{Z\to \infty}N_c(Z)/Z=1.$
\end{itemize}
\end{theorem}

In particular, the spectrum of hydrogen ($N=1$) is explicitly known \cite{Y91} (see also \cite{PP02} for a review).

\begin{theorem}[Hydrogen spectrum]\label{thm:hydrogen-spectrum-original}
All negative eigenvalues of the operator $-\frac{1}{2}\Delta-|x|^{-1}$ in $L^2(\R^2)$ are
\[
E_n  =  - \frac{1}{{2(n + 1/2)^2 }},
\]
with multiplicity $2n+1$, where  $n=0,1,2,...$
\end{theorem}

The following consequence will be useful in our estimates. The proof can be found in the Appendix.

\begin{lemma}[Hydrogen semiclassics]\label{le:hydrogen-1} When $\mu\to 0^+$ we have
\bq \label{eq:asymtotic-hydrogen-mu}
\Tr \left[ { - \frac{1}
{2}\Delta  - |x|^{ - 1}  + \mu } \right]_ -   = \frac{1}
{2} \left[ {  \ln (\mu ) - 3\ln (2) - 2\gamma_{E} + 1 } \right]+ o(1).
\eq
By scaling, for $\mu>0$ fixed and $h\to 0^+$, 
\bq \label{eq:asymtotic-hydrogen-h}
  \Tr \left[ { - h^2\Delta  - |x|^{ - 1}  + \mu} \right]_ -   =(4h^2)^{ -1} \left[ {\ln (2h^2) +\ln (\mu ) +c^{\mathrm{H}}} \right] + o(h^{ - 2} ),
\eq
where $c^{\mathrm{H}}  =  - 3\ln (2) - 2\gamma _E  + 1 \approx -2.2339$ with $\gamma_{E}  \approx 0.5772$ being Euler's constant \cite{Eu1736}.
\end{lemma}

\subsection{Useful Inequalities}\label{subsec:Useful-Inequalities}
For the readers' convenience, we recall some usual notations. We shall denote by $L^2(\R^2)$ the Hilbert space with the inner product $(f,g)=\int_{\R^2}\overline{f(x)}g(x)dx$. An operator $\gamma$ on $L^2(\R^2)$ is called a (one-body) {\it density matrix} if $0\le \gamma\le 1$ and $\Tr(\gamma)<\infty$. Its {\it density} is $\rho_{\gamma}(x):=\gamma(x,x)$, where $\gamma(x,y)$ is the kernel of $\gamma$. More precisely, if $\gamma$ is written in the spectral decomposition $\gamma= \sum_i t_i \left| {u_i} \right\rangle \left\langle {u_i} \right|$ then $\gamma(x,y):= \sum_i t_i u_i(x) \overline{u_i(y)}$ and $\rho_\gamma(x):=\sum_i t_i |u_i(x)|^2.$ For example, the density matrix $\gamma_\Psi$ of a (normalized) wave function $\Psi\in \bigwedge_{i = 1}^N L^2 (\mathbb{R}^2)$ is 
\[
\gamma _\Psi  (x,y) := N\int_{\R^{2(N-1)}}{\Psi (x,x_2 ,...,x_N )\overline {\Psi (y,x_2 ,...,x_N )} dx_2 ...dx_N } ,
\]
which satisfies $0\le \gamma_\Psi \le 1$ and $\Tr(\gamma_\Psi)=N$. Moreover, its density  is 
\[
\rho_\Psi(x) := \rho_{\gamma_\Psi}(x)= N\int _{\R^{2(N-1)}}{|\Psi (x,x_2 ,...,x_N )|^2 dx_2 ...dx_N }.
\]

The following theorem regarding the spectrum of Schr\"odinger operators is important for our analysis (see e.g. \cite{LL01} for a proof). The analogue in three dimensions was  first proved by Lieb and Thirring \cite{LT75}.

\begin{theorem}[Lieb-Thirring inequalities]\label{thm:LT-inequality}
There exists a finite constant $L_{1,2}>0$ such that for any real-valued potential $V$ with $V_+ \in L^{2}(\R^2)$ one has
\bq\label{eq:LT-eigenvalue-sum}
	\Tr[-\Delta -V]_-\ge -L_{1,2} \int_{\R^2} {V_+^{2}(x)dx},
\eq
where $a_+ := \max\{a,0\}$ and $a_- := \min\{a,0\}$. Hence $\Tr[-\Delta -V]_-$ is the sum of all negative eigenvalues of $-\Delta-V$ in $L^2(\R^2)$.

Equivalently, there exists a finite constant $K_{2}>0$ such that for any density matrix $\gamma$ one has
\bq \label{eq:LT-kinetic}
	\Tr [-\Delta \gamma] \ge K_2 \int_{\mathbb{R}^2 } {\rho _{\gamma}^2 (x)dx} .
\eq
\end{theorem}

Note that in general there is no upper bound on $\Tr 1_{(-\infty,0]}(-\Delta -V)$, the number of negative eigenvalues of $-\Delta-V$, in term of $\int V_+^\alpha$ for any $\alpha>0$. However, we shall only need some localized versions of this bound. The proof of the following lemma can be found in the Appendix. The estimate in (ii) is useful to treat the Coulomb singularity in the region close to the origin (recall that $|x|^{-1}\notin L^2_{\loc}(\R^2)$).

\begin{lemma} \label{le:hydrogen-2} Let $V:\R^2\to \R$ and let $0\le \phi(x) \le 1$ supported in a subset $\Omega\subset \R^2$ with finite measure $|\Omega|$. Let $0\le \gamma \le 1$ be an operator on $L^2(\R^2)$ such that 
$$\Tr[(- h^2 \Delta  -V)\phi \gamma \phi]\le 0~\text{for some}~1/2>h>0.$$

\begin{itemize} 

\item[$(i)$] If $V_+\in L^2_{\loc}(\R^2)$ then $\phi \gamma \phi$ is trace class and there exists a universal constant $C>0$ (independent of $V$, $\gamma$ and $h$) such that for any $\alpha \in [0,1]$,
$$\int_{\R^2}\rho_{\phi\gamma\phi}^{2\alpha}(x)dx \le C h^{-4\alpha} ||V_+||_{L^2(\Omega)}^{2\alpha}|\Omega|^{1-\alpha}.$$

\item[$(ii)$] If  $V(x)\le C_0(|x|^{-1}+1)$ then $\phi \gamma \phi$ is trace class and there exists a constant $C>0$ dependent only on $C_0$ (but independent of $V$, $\phi$, $\Omega$, $\gamma$ and $h$) such that for any $\alpha\in [0,1]$,
$$\int_{\R^2}\rho_{\phi\gamma\phi}^{2\alpha}(x)dx \le Ch^{-4\alpha} \left( {|\ln h| + |\Omega| } \right)^{\alpha}|\Omega|^{1-\alpha}.$$

\end{itemize}
\end{lemma}

We shall approximate the ground state energy $E(N,Z)$ by one-body densities. For the lower bound, we need the following inequality to control the electron-electron repulsion energy. The three-dimensional analogue of this bound was first proved by Lieb \cite{Li79} and was then improved by Lieb and Oxford \cite{LO81}. The two-dimensional version below was taken from \cite{LSY95}.

\begin{theorem}[Lieb-Oxford inequality]\label{thm:Lieb-Oxford-inequality} 
For any (normalized) wave function  $\Psi \in\bigwedge_{i = 1}^N L^2 (\mathbb{R}^2)$ it holds that
\[
\left( {\Psi, \sum\limits_{1 \le i < j \le N} {\frac{1}
{{|x_i  - x_j |}}} \Psi } \right)  \ge D(\rho _{\Psi  } ) - C_{\mathrm{LO}}\int {\rho _{\Psi  } ^{3/2} },
\]
with $C_{\mathrm{LO}} = 192(2\pi)^{1/2}$, where the direct term $D(\rho_{\Psi} )$ is defined as in Section \ref{subsec:Coulomb-Potential}.
\end{theorem}

For the upper bound, we shall need the next result \cite{Li81B}.

\begin{theorem}[Lieb's variational principle] \label{thm:Lieb-variational-principle} For $Z>0$, $N\in \mathbb{N}$ and any density matrix $\gamma$ with $\Tr(\gamma)\le N$, one has
\[
	E(N,Z) \le \Tr \left[ {\left( { - \frac{1}{2}\Delta  - Z|x|^{ - 1} } \right)\gamma } \right] + D(\rho_{\gamma} ) 
	- \frac{1}{2}\iint {\frac{{|\gamma (x,y)|^2 }}{{|x - y|}}dxdy},
\]
where the direct term $D(\rho_{\gamma} )$ is defined as in Section \ref{subsec:Coulomb-Potential}.
\end{theorem} 

\subsection{Coulomb Potential} \label{subsec:Coulomb-Potential}
Here we study the Coulomb potential $f*\clmb^{-1}$ of some function $f$. Associated to this potential is the Coulomb energy of two functions, 
\[
D(f,g): = \frac{1}
{2}\iint\limits_{\mathbb{R}^2  \times \mathbb{R}^2 } {\frac{{\overline{f(x)}g(y)}}
{{|x - y|}}dxdy}.
\]
That $D(f,g)$ is well-defined at least in $L^{4/3}(\R^2)$ is due to the Hardy-Littlewood-Sobolev inequality (see \cite{LL01}, Theorem 4.3)
$$D(|f|,|g|) \le C_{\mathrm{HLS}} \left\| f \right\|_{L^{4/3} }\left\| g \right\|_{L^{4/3} } ~~\text{for all}~f,g\in L^{4/3}(\R^2).$$

Moreover, $|x-y|^{-1}$ is a strictly positive kernel since the 2D Fourier transform of $\clmb^{-1}$ is itself up to a constant (see \cite{LL01} Theorem 5.9). Therefore, $D(f) := D(f,f)$ is always nonnegative and $(f,g)\mapsto D(f,g)$ is a positive inner product in $L^{4/3}(\R^2)$. These observations allow us to formulate the following theorem.

\begin{theorem}[Coulomb norm]\label{thm:Coulomb-norm} There exists $C_{\mathrm{HLS}}$  such that
\[
0<D(f) \le C_{\mathrm{HLS}} \left\| f \right\|_{L^{4/3} }^2 ~~\text{for all}~f\in L^{4/3}(\R^2)\minus \{0\}.
\]
Consequently, $f\mapsto \sqrt{D(f)}$ is a norm in $L^{4/3}(\R^2)$. 
\end{theorem}

In three dimensions, the Coulomb potential $\rho*\clmb^{-1}$ of a radially symmetric function $\rho$ is represented beautifully by Newton's Theorem (see \cite{LL01}, Theorem 9.7). In two dimensions, however, we do not have such a representation since $\clmb^{-1}$ is not the fundamental solution to the two-dimensional Laplace operator. Therefore, the following bounds will be useful in our context and their proofs can be found in the Appendix. The lower bound is similar to Newton's Theorem in three dimensions, but the upper bounds are more involved. We do not claim that they are optimal but they are sufficient for our purposes. 

\begin{lemma}[Coulomb potential bound]\label{le:convolution-bound} 
Assume that $\rho$ is radially symmetric, $0\le \rho(x)\le (2\pi |x|)^{-1}$ and $\int \rho =\lambda$. We have the following bounds on the potential $\rho*\clmb^{-1}$.
\begin{itemize}
	\item[$(i)$] (Lower bound) For all $x\in \R^2\minus\{0\}$, 
	$$
	(\rho *\clmb^{ - 1} )(x) \ge \int\limits_{\mathbb{R}^2 } {\frac{{\rho (y)}}
	{{\max \{ |x|,|y|\} }}dy}.
	$$
	
	\item[$(ii)$] (Upper bound) For all $x\in \R^2\minus\{0\}$, 
	$$
	(\rho*\clmb^{ - 1} )(x) \le 2\sqrt{2 \lambda}|x|^{ - 1/2}  +3.
	$$
	Moreover, for any $\delta>0$ there exists $R=R(\rho,\delta)>0$ and a universal constant $C_1>0$ such that for any $|x|\ge R$, 
	$$
	(\rho *\clmb^{ - 1} )(x) \le \frac{{\lambda+ \delta }}
	{{|x|}}+C_1\frac{{\ln(|x|)}}
	{{|x|}}\int_{3|x|/2 \ge |y| \ge |x|/2} {\rho (y)dy}.
	$$
\end{itemize}
\end{lemma}

\section{Thomas-Fermi Theory}\label{sec:TF-theory}
In this section, we introduce the two-dimensional Thomas-Fermi (TF) theory which will turn out to be the main tool to understand the ground state energy and ground states. The three-dimensional TF theory was studied in great mathematical detail by Lieb-Simon \cite{LS77,Li81}. In fact, the simplest version of TF theory (see \cite{LL01}, Chap. 11) is sufficient for our discussion here.

\begin{definition}[Thomas-Fermi functional]
For any nonnegative function $\rho\in L^1(\R^2)$ we define the {\it TF functional} as
\bqq
  \E^{\TF}(\rho ) &: =& \int\limits_{\mathbb{R}^2 } {\left( {\pi \rho ^2 (x) - \frac{{\rho (x)}}
{{|x|}} + (4\pi)^{-1}[|x|^{ - 1}  - 1]_ + ^2 } \right)dx}  + D(\rho )  .
\eqq
For any $\lambda>0$ we define the {\it TF energy} as
\bq \label{eq:TF-variational-problem}
E^{\TF} (\lambda ): = \inf \left\{ {\E^{\TF}(\rho )|\rho  \ge 0,\left\| \rho  \right\|_{L^1 (\mathbb{R}^2 )}  \le \lambda } \right\}.
\eq
\end{definition}

\begin{remark}

\begin{itemize}

\item[(i)] The term $\pi \rho^2$ comes from the semiclassics of the kinetic energy while $-\int \rho(x)|x|^{-1}$ and the direct term $D(\rho)=\frac{1}{2}\iint \rho(x)\rho(y)|x-y|^{-1}dxdy$ stand for the Coulomb interactions.

\item[(ii)] The appearance of $(4\pi)^{-1}[|x|^{ - 1}  - 1]_ + ^2$ ensures that the TF functional is bounded from below. In fact, 
\bq\label{eq:TF-functional-2def}
	\E^{\TF}(\rho )  &=& \int\limits_{|x| \le 1} {\pi \left( {\rho (x) - \frac{1}{{2\pi |x|}}} \right)^2 dx}  + \int\limits_{|x| > 1} {\left( {\pi \rho ^2 (x) - \frac{{\rho (x)}}
	{{|x|}}} \right) dx}+ D(\rho )  - \frac{3}{4}\nn\\
&\ge & -\int \rho - \frac{3}{4} \nn.
\eq

\item[(iii)] If $\rho^{\TF}_\lambda$ is the ground state of the above TF theory then $Z\rho^{\TF}_\lambda$ is expected to approximate the density $\rho_{\Psi_{N,Z}}$ of a ground state $\Psi_{N,Z}$ of $H_{N,Z}$ with $N \approx \lambda Z$ (in some appropriate sense). In other words, the $Z$-dependent TF theory can be defined from the above TF theory by the scaling $\rho \mapsto Z \rho$. 
\end{itemize}
\end{remark}
Basic information about the TF theory is collected in the following theorem.

\begin{theorem}[Thomas-Fermi theory] \label{thm:TF-theory}  
Let $\lambda>0$.
\begin{itemize}
	\item [$(i)$] (Existence)  The variational problem \eqref{eq:TF-variational-problem} has a unique minimizer $\rho^{\TF}_\lambda$. Moreover, the 
	functional $\lambda\mapsto E^{\TF} (\lambda )$ is strictly convex, decreasing on $(0,1]$ and $E^{\TF}(\lambda)=E^{\TF}(1)$ if $\lambda\ge 1$.
	
	\item [$(ii)$] ({\rm TF} equation)  $\rho^{\TF}_\lambda$ satisfies the TF equation
	$$ 2\pi \rho^{\TF}_\lambda(x) =\left[ {|x|^{-1}-(\rho^{\TF}_\lambda*\clmb^{-1})(x)-\mu^{\TF}_\lambda} \right]_+$$
	with some constant $\mu^{\TF}_\lambda>0$ if $\lambda<1$ and $\mu^{\TF}_\lambda=0$ if $\lambda\ge 1$. 
	
	\item [$(iii)$] ({\rm TF} minimizer) $\rho^{\TF}_\lambda$ is radially symmetric; $\int \rho^{\TF}_\lambda=\min\{\lambda,1\}$ and
	\bqq
	0 \le |x|^{-1} - 2\pi\rho^{\TF}_{\lambda} \le C |x|^{-1/2}~~\text{for all}~x \neq0. 
	\eqq
	Moreover, $\supp \rho^{\TF}_\lambda$ is compact if and only if $\lambda<1$.
\end{itemize}
\end{theorem}

\begin{remark}
Henceforth we shall always denote by $C$ some finite positive constant depending only on $\lambda>0$ (the total mass in the TF theory). Two $C$'s in the same line may refer to two different constants. 
\end{remark}

\begin{proof} (i-ii) Formula \eqref{eq:TF-functional-2def} implies that $\rho \mapsto \E^{\TF}(\rho)$ is strictly convex. Therefore, the existence and uniqueness of the TF minimizer, and  the TF equation follow from standard variational methods similarly to the three-dimensional TF theory (see \cite{LL01}, Theorems 11.12 and 11.13). The property of $\mu^{\TF}_\lambda$ is a consequence of the TF equation and is shown in Lemma~\ref{le:TF-equation} below.

That $E^{\TF}(\lambda)$ is decreasing follows from the definition. When $\lambda \ge 1$, $E^{\TF}(\lambda)=E^{\TF}(1)$ since $\rho^{\TF}_\lambda=\rho^{\TF}_1$ (by (iii)). When $\lambda \in (0,1]$, the TF energy is also strict convex because the unique minimizer satisfies $\int \rho^{\TF}_\lambda=\lambda$ (by (iii)) and the TF functional is strict convex.  

(iii) Since the TF functional is rotation invariant and the minimizer is unique, it must be radially symmetric. The inequality $0 \le |x|^{-1} - 2\pi\rho^{\TF}_{\lambda} \le C |x|^{-1/2}$ follows from the TF equation and the following estimate in Lemma~\ref{le:convolution-bound},
$$
(\rho ^{\TF} *\clmb^{ - 1} )(x) \le 2\sqrt{2 \lambda}|x|^{ - 1/2}+3.
$$
We defer the proof that $\int \rho^{\TF}_\lambda=\min\{\lambda,1\}$ and property of $\supp \rho^{\TF}_\lambda$ to Lemma~\ref{le:TF-equation}.
\end{proof}

\subsection{Thomas-Fermi Equation}

\begin{lemma}[TF equation]\label{le:TF-equation} Assume that $\rho$ is a nonnegative, radially symmetric,  integrable  solution to the TF equation
\bq \label{eq:TF-equation-rho}
	2\pi \rho(x)=\left[ {|x|^{-1}-(\rho*\clmb^{-1})(x)-\mu} \right]_{+}
\eq
for some constant $\mu\ge 0$. 
\begin{itemize}
\item[$(i)$] If $\mu>0$ then $\int \rho <1$ and $\supp \rho$ is compact.

\item[$(ii)$] If $\mu=0$ then $\int \rho =1$ and 
\[
\int_{|x| \ge r} {\rho (x)dx}  \ge e^{-2\sqrt{r}}~~\text{for all}~r\ge 0.
\]
\end{itemize}
\end{lemma}

\begin{proof} Denote $\int \rho =:\lambda>0$. For $r>0$ we shall write $\rho(r)$ instead of $\rho(x)|_{|x|=r}$. 

1. We start by proving $\lambda\le 1$. Since $\rho$ is nonnegative and radially symmetric, we have by Lemma~\ref{le:convolution-bound} 
$$
(\rho *\clmb^{ - 1} )(x) \ge \int_{\mathbb{R}^2 } {\frac{{\rho (y)}}
{{\max \{ |x|,|y|\} }}dy}.
$$
Hence, the TF equation \eqref{eq:TF-equation-rho} yields   
\bq\label{eq:upper-bound-xrho}
2\pi \rho (x)|x| \le \left[ {1 - \int_{\mathbb{R}^2 } {\frac{{|x|\rho (y)}}
{{\max \{ |x|,|y|\} }}dy} -\mu |x|} \right]_ + ~~\text{for all}~x\ne 0 .
\eq
For any $\eps\in (0,\lambda)$, we can find $R_\eps>0$ such that $\int_{|x|\ge R_\eps} \rho =\eps$. When $|x|\ge R_\eps$, using
$$
\int_{\mathbb{R}^2 } {\frac{{|x|\rho (y)}}{{\max \{ |x|,|y|\} }}dy}\ge  \int_{|y|\le R_\eps} \rho (y) =\lambda -\eps
$$
we can deduce from \eqref{eq:upper-bound-xrho} that
$$ 2\pi \rho(x)|x|\le [1-\lambda+\eps - \mu |x|]_+ \le [1-\lambda+\eps - \mu R_\eps]_+~~\text{for all}~|x|\ge R_\eps.$$
Since $\int_{|x|\ge R_\eps} \rho =\eps>0$, there exists $|x|\ge R_\eps$ such that $\rho(x)>0$. Therefore, it follows from the latter estimate that 
\bq \label{eq:1-lambda-eps-muR}
	1-\lambda+\eps - \mu R_\eps\ge 0~~\text{for all}~\eps\in (0,\lambda).
\eq
For any $\mu\ge 0$, \eqref{eq:1-lambda-eps-muR} implies that $\lambda\le 1$. 

2. If $\mu>0$ then \eqref{eq:1-lambda-eps-muR} yields
$$\limsup_{\eps\to 0} R_\eps \le R_0:=\mu^{-1}(1-\lambda).$$
Since $\int_{|x|\ge R_\eps} \rho =\eps$ and $\limsup_{\eps\to 0} R_\eps \le R_0$, we get $\int_{|x|\ge R_0} \rho =0$. Thus $\supp \rho \subset \{|x|\le R_0\}$ and $\lambda<1$ (because $R_0>0$).  

3. From now on we assume that $\mu=0$. We shall prove that $\lambda=1$. Suppose that $\lambda<1-3\eps$ for some $\eps>0$. Because $\rho$ is nonnegative, radially symmetric and $\rho(x)\le (2\pi|x|)^{-1}$ (due to the TF equation \eqref{eq:TF-equation-rho}), by Lemma~\ref{le:convolution-bound} we can find $R>0$ and $C_1>0$ such that
\bq\label{eq:prebound-convolution}
	(\rho *\clmb^{ - 1} )(x) \le \frac{{1-2\eps }}{{|x|}}+C_1\frac{{\ln(|x|)}}{{|x|}}\int_{3|x|/2 \ge |y| \ge |x|/2} {\rho (y)dy}~~\text{for all}~|x|\ge R.
\eq
Define $\eps_1:=\eps/C_1$ and 
\[
A := \left\{ {r\ge R:\int_{3r/2 \ge |y| \ge r/2} {\rho (y)dy}  \le \frac{{\eps_1 }}
{{\ln(r^{ - 1})}}} \right\}.
\]
If $|x|\in A$, then \eqref{eq:prebound-convolution} gives $(\rho *\clmb^{ - 1} )(x)\le (1-\eps)|x|^{-1}$, and the TF equation \eqref{eq:TF-equation-rho} with $\mu=0$ gives
\[
2\pi \rho (x) = \left[ {\frac{1}{{|x|}} - (\rho *|x|^{ - 1} )(x)} \right]_ +   \ge \frac{\eps}{{|x|}}.
\]
Taking the integral of the previous inequality over $\{x\in \R^2:|x|\in A\}$ one has
\[
	\infty>2\pi \int\limits_{\R^2} \rho \ge 2\pi \int\limits_{|x| \in A} \rho(x)dx \ge \int\limits_{|x| \in A} {\frac{\eps}{{|x|}}dx} = 2\pi \eps \mathcal{L}^1 (A)
\]
where $\mathcal{L}^1$ is the one-dimensional Lebesgue measure. Thus $A$ has finite measure, and consequently we can choose a sequence $\{R_n\}_{n=1}^\infty \subset \R\minus A$ such that $3R<3R_{n}<R_{n+1}<4R_{n}$ for all $n\ge 1$. Because $R_n>R$ and $R_n\notin A$ we have, by the definition of $A$, 
\[
\int\limits_{3R_n /2 \ge |y| \ge R_n /2} {\rho (y)dy}  >  \frac{{\eps _1 }}
{{\ln(R_n)}}~~\text{for all}~n\ge 1.
\]
Taking the sum over all $n\in \mathbb{N}$ and using $R_{n+1}>3R_n$, we find that  
$$
\infty > \int\limits_{\mathbb{R}^2 } {\rho}  \ge \sum\limits_{n = 1}^\infty  {\int\limits_{3R_n /2 \ge |y| \ge R_n /2} {\rho (y)dy} }  \ge  \sum\limits_{n = 1}^\infty  {\frac{{\eps _1 }}
{{\ln(R_n)}}} .
$$
On the other hand, since  $R_{n+1}<4R_n$ we get $R_n\le 4^n R_1\le [4(1+R_1)]^n$ for all $n\ge 1$. Therefore,
$$
\sum\limits_{n = 1}^\infty  {\frac{{\eps _1 }}
{{\ln (R_n)}}}  \ge \sum\limits_{n = 1}^\infty  {\frac{{\eps _1 }}
{{n\ln (4(1 + R_1 ))}}} =+ \infty 
$$
The last two inequalities yield a contradiction.

4. Finally, we show the lower bound on $\int_{|x|\ge r}\rho$. With $\mu=0$ and $\lambda=1$, inequality \eqref{eq:upper-bound-xrho} becomes 
\bq \label{eq:upper-bound-xrho-v1}
2\pi \rho (x)|x| \le \int\limits_{|y| \ge |x|} {\left( {1 - \frac{{|x|}}
{{|y|}}} \right)\rho (y)dy} ~~\text{for all}~x\ne 0.
\eq
Denote
\[
g(r): = \int\limits_{|y| \ge r} {\left( {1 - \frac{r}
{{|y|}}} \right)\rho (y)dy}  = 2\pi \int\limits_r^\infty  {\left( {s - r} \right)\rho (s)ds} .
\]
Then $g(0)=1$, $g(+\infty)=0$ and
\[
g'(r) =  - 2\pi \int\limits_r^\infty  {\rho (s)ds}  < 0,~~g''(r) = 2\pi \rho (r)~~\text{for all}~r>0.
\]
Thus \eqref{eq:upper-bound-xrho-v1} can be rewritten as 
\[
rg''(r) \le g(r)~~\text{for all}~r>0.
\]
Note that $g_0(r):=e^{-2\sqrt{r}}$ satisfies $g_0(0)=1$, $g_0(+\infty)=0$ and 
$$rg_0''(r)-g_0(r)= \frac{1}{2\sqrt{r}}e^{-2\sqrt{r}}>0.$$
Therefore, $h(x):=g(x)-g_0(x)$ satisfies that $h(0)=h(+\infty)=0$ and $rh''(r)\le h(r)$. If the set $U:=\{r>0: h(r)<0\}$ is not empty, then $h$ is a strict concave function on this open set. By the maximum principle and $h(0)=h(+\infty)=0$, we can argue to get a contradiction. Thus $h(r) \ge 0$ for all $r \ge 0$. This yields $\int_{|x| \ge r} \rho(x)dx \ge g(r) \ge g_0(r) = e^{-2\sqrt{r}}$.
\end{proof}

\section{Semiclassics for the TF Potential}\label{sec:Semiclassics-TF-potential}
In this section, we consider the semiclassics for the TF potential
$$V^{\TF}_\lambda(x) :=|x|^{-1}-(\rho^{\TF}_\lambda*\clmb^{-1})(x)-\mu^{\TF}_\lambda.$$
From the TF equation and the properties of the TF minimizer (see Theorem~\ref{thm:TF-theory}) we have $[V^{\TF}_\lambda]_+ \in L^1(\R^2)$ and 
$$|V^{\TF}_{\lambda}(x)- |x|^{-1} |  \le C(|x|^{-1/2}+1)~~\text{for all}~x\ne 0.$$

The following theorem will turn out to be the main ingredient to prove Theorems \ref{thm:GSE} and \ref{thm:radius-atom}. The parameter $h$ will eventually be replaced by $(2Z)^{-1/2}$ in our application.

\begin{theorem}[Semiclassics for TF potential]\label{thm:semiclassical-approximation} 
When $h\to 0^+$ one has
\bq\label{eq:semiclassical-approximation}
	\Tr \left[ { -h^2\Delta  - V^{\TF}_\lambda} \right]_ -   &=& -(8\pi h^2)^{-1}\int_{\mathbb{R}^2 } 
	{\left( {[V^{\TF}_\lambda(x)]_ + ^2  - [|x|^{ - 1}  - 1 ]_ + ^2 } \right)dx}\nn  \hfill \\
	&~&+ (4h^2)^{ -1} \left[ {\ln(2h^2) +c^{\mathrm{H}}} \right] + o(h^{ - 2} ) .
\eq
where $c^{\mathrm{H}}  =  - 3\ln(2) - 2\gamma _E  + 1 \approx -2.2339$. 

Moreover, there is a density matrix $\gamma_h$ such that 
\bq \label{eq:trial-density-ineq1}
\Tr \left[ { (-h^2\Delta  - V^{\TF}_\lambda) \gamma_h} \right] = \Tr \left[ { -h^2\Delta  - V^{\TF}_\lambda} \right]_ - + o(h^{-2})
\eq
and 
\bq\label{eq:trial-density-ineq2}
	2h^2\Tr(\gamma_h) \le \int \rho^{\TF}_\lambda, \quad D((2h^2) \rho_{\gamma_h}-\rho^{\TF}_\lambda) =o(1).
\eq
\end{theorem}

Note that (\ref{eq:semiclassical-approximation}) is a special case of Theorem~\ref{thm:semiclassics-general-potential}. In this section, we shall prove (\ref{eq:semiclassical-approximation}) in detail. The proof of Theorem~\ref{thm:semiclassics-general-potential} is provided in the next section. 

As in \cite{SS03} we shall prove the semiclassical approximation \eqref{eq:semiclassical-approximation} by comparing with the hydrogen. In fact, because of the hydrogen semiclassics \eqref{eq:asymtotic-hydrogen-h}, the approximation \eqref{eq:semiclassical-approximation} is equivalent to
\bq \label{eq:hydrogen-comparison}
  &~&\Tr \left[ { -h^2\Delta  - V^{\TF}_\lambda } \right]_ -   - \Tr \left[ { -h^2\Delta  - |x|^{ - 1}  + 1} \right]_ - \nn \hfill \\
   &=&-(8\pi h^2)^{-1}\int_{\mathbb{R}^2 } {\left( {[V^{\TF}_\lambda]_ + ^2  - [|x|^{ - 1}  -1 ]_ + ^2 } \right)dx}  + o(h^{-2}).
\eq

\subsection{Localization}
To treat the singularity of the TF potential we shall distinguish between three regions. In the interior region (close to the origin), we shall compare directly with hydrogen; while in the exterior region (not too close and not too far from the origin) we can employ the usual semiclassical techniques; and finally, the region very far from the origin has negligible contribution. 

\begin{definition}[Partition of unity]\label{def:partition-unity}  
Let $\varphi$ be a nonnegative, smooth function (with bounded derivatives) such that $\varphi(x)=1$ if $|x|\le 1$ and $\varphi(x)=0$ if $|x|\ge 2$. Choose $r:=h^{1/2}$, $\Lambda:=|\ln h|$ and denote 
\bqq
	\Phi _1 (x) &=& \varphi(x/r), \hfill\\
	\Phi _ {2}(x) &=& (1-\varphi^2(x/r))^{1/2}\varphi (x/\Lambda ), \hfill\\
	\Phi_3(x) &=& (1-\varphi^2(x/\Lambda ))^{1/2}.
\eqq
Then $\sum_{i=1}^3\Phi _ i^2=1$, $\supp \Phi_1\subset \{|x|\le 2r\}$, $\supp \Phi_2 \subset \{r\le |x|\le 2\Lambda  \}$, $\supp \Phi_3 \subset \{|x|\ge \Lambda \}$.
\end{definition} 

The localization cost is controlled by the following lemma. 

\begin{lemma}[Localization]\label{le:localization} Let $V$ be either $V^{\TF}_{\lambda}$ or $(|x|^{-1}-1)$. When $\Lambda=|\ln h|$ and $r=h^{1/2}\to 0^+$ one has
\bqq
	\Tr [ - h^2 \Delta  - V ]_ - &=& \sum_{i=1,2} \Tr [\Phi_i ( - h^2 \Delta  - V)\Phi_i ]_ -+o(h^{-2})
\eqq
\end{lemma}
Note that in the sum on the right-hand side the contribution of region $\supp \Phi_3$ does not appear.

\begin{proof} 
1. To prove the lower bound, using the IMS formula 
$$ - \Delta  =  \sum_{i=1}^3 \Phi_i (-\Delta -u) \Phi_i~~\text{with}~u:=\sum_{i=1}^3 { |\nabla \Phi _i|^2}\le Cr^{-2}1_{\{|x|\le 2\Lambda\}}$$
one has
\bqq
  \Tr [ - h^2 \Delta  - V]_ -  \ge  \sum\limits_{i = 1}^3  \Tr [\Phi _i ( - h^2 \Delta  - V -
Ch^2 r^{ - 2} 1_{\{ |x| \le 2\Lambda \} } 
)\Phi _i]_-.
\eqq
The term involving $\Phi_3$ has negligible contribution. Indeed, since $\supp \Phi_3 \subset \{|x|\ge \Lambda\}$, it follows from the Lieb-Thirring inequality \eqref{eq:LT-eigenvalue-sum} that
\bqq
 	&~& \Tr [\Phi _3 ( - h^2 \Delta  - V - Ch^2 r^{ - 2} 1_{\{ |x| \le 2\Lambda \} } )\Phi _3 ]_ - \\
 	&\ge & \Tr [ - h^2 \Delta  - 1_{\{ |x| \ge \Lambda \} } (V_+ + Ch^2 r^{ - 2} 1_{\{ |x| \le 2\Lambda \} } )]_ - \\
	&\ge&  - L_{1,2} h^{ - 2} \int_{|x| \ge \Lambda } {[V_+(x) + Ch^2 r^{ -2} 1_{\{ |x| \le 2\Lambda \} } ]^2 dx}  = o(h^{ - 2} ).
\eqq
Here note that $\lim_{\Lambda \to \infty} \int_{|x|\ge \Lambda} V_+^2 = 0$ since $1_{\{|x|\ge 1\}}V_+\in L^2(\R^2)$ and $h^4r^{-4} \int_{|x|\le 2\Lambda} \to 0$ because $\Lambda = |\ln h|$.

Moreover, for $i=1,2$, if we denote 
\[
\gamma _i : = 1_{( - \infty ,0]} \left( {\Phi _i ( - h^2 \Delta  - V - Ch^2 r^{ - 2} )\Phi _i } \right)
\]
then $\Tr(\Phi_i \gamma_i \Phi_i)\le Ch^{-2}(|\ln h|+\Lambda^2)$ by Lemma~\ref{le:hydrogen-2} (ii). Therefore,  
\bqq
\Tr [\Phi _i ( - h^2 \Delta  - V - Ch^2 r^{ - 2} )\Phi _i ]_ - &=& \Tr [\Phi _i ( - h^2 \Delta  - V- Ch^2 r^{ - 2})\Phi _i \gamma_i] \\
&=& \Tr [\Phi _i ( - h^2 \Delta  - V)\Phi _i \gamma_i] -Ch^2 r^{ - 2}\Tr(\Phi_i \gamma_i \Phi_i) \\
&\ge & \Tr [\Phi _i ( - h^2 \Delta  - V)\Phi _i ]_- +o(h^{-2}).
\eqq

2. To show the upper bound, we choose
\bqq
\gamma ^{(i)} :=1_{( - \infty ,0]} \left( {\Phi_i ( - h^2 \Delta  - V)\Phi_i } \right),~ \gamma^{(0)} := \sum_{i=1,2} \Phi_i \gamma _i \Phi_i.
\eqq
Since $0\le \gamma ^{(i)}\le 1$ ($i=1,2$) and $\sum_{i=1,2} \Phi_i^2\le 1$ we have $0\le \gamma^{(0)} \le 1$. Thus,
\bqq
	\Tr [ - h^2 \Delta  - V]_ -   \le  \Tr [( - h^2 \Delta  - V )\gamma ^{(0)} ]   = \sum\limits_{i = 1,2} {\Tr [\Phi _i ( - h^2 \Delta  - V)\Phi _i ]_ -  }.
\eqq
\end{proof}

\subsection{Hydrogen Comparison in Interior Region}
In the interior region, we shall compare the semiclassics of the TF potential directly with hydrogen. Note that 
$$\left| { (8\pi h^2)^{-1}\int {[V^{\TF}_{\lambda}]_+ ^2(x)-[|x|^{-1}-1]_+^2\Phi _1^2 (x)^2dx}} \right|\le Crh^{-2}=o(h^{-2})$$
because $|V^{\TF}_{\lambda}-|x|^{-1}|\le C(|x|^{-1/2}+1)$ and $\supp \Phi_1 \subset \{|x|\le 2r\}$. This inequality is the semiclassial version of the following bound.
 
\begin{lemma}[Hydrogen comparison in interior region]\label{le:Hydrogen-comparison-interior-region} When $r = h^{1/2}\to 0$ we have    
\bqq
&~&\Tr \left[ {\Phi _1 \left( { - h^2\Delta -V^{\TF}_\lambda } \right)\Phi _1 } \right]_ -  -\Tr \left[ {\Phi _1 \left( { - h^2\Delta  - |x|^{-1} + 1} \right)\Phi _1 } \right]_ - = o(h^{-2}).
\eqq
\end{lemma}
\begin{proof} 
The lower and upper bounds can be proved in the same way. We prove for example the upper bound. If we denote 
\[
\gamma^{(1)} : = 1_{( - \infty ,0]} \left[ {\Phi _1 \left( { - h^2 \Delta  - |x|^{-1}+1} \right)\Phi _1 } \right]
\]
then by Lemma~\ref{le:hydrogen-2} (ii),
\bq \label{eq:Tr-Phi1-gamma1-Phi1}
	\Tr[\Phi_1 \gamma ^{(1)}\Phi_1] \le Crh^{-2}|\ln h|^{1/2}.
\eq 
By using $|V^{\TF}_{\lambda}(x)-|x|^{-1}+1|\le C(|x|^{-1/2}+1)\le Cr^{-1/2}$ for $x \in \supp \Phi_1$ we get
\bqq
	\Tr \left[ {\Phi_1\left( { - h^2 \Delta  - |x|^{-1}+1 } \right)\Phi_1 } \right]_ -   &=& \Tr \left[ {\Phi_1 \left( { - h^2 \Delta  - |x|^{-1}+1 } \right)\Phi_1
	\gamma ^{(1)} } \right] \\
	&\ge& \Tr \left[ {\Phi_1 \left( { - h^2 \Delta  - V_\lambda^{\TF}} \right)\Phi_1 \gamma _1 } \right] - Cr^{-1/2}\Tr[\Phi_1 \gamma ^{(1)}\Phi_1]  \\
	&\ge& \Tr \left[ {\Phi_1 \left( { - h^2 \Delta  - V_\lambda^{\TF}} \right)\Phi_1 } \right]_ -  +o(h^{-2}).
\eqq
\end{proof}
 
\subsection{Semiclassics in Exterior Region}
In the exterior region, the standard semiclassiccal technique of using coherent states \cite{Li81,Th81} (see also \cite{LL01,So03}) is available. 
\begin{definition}[Coherent states]\label{def:coherent}
Let $g$ be a radially symmetric, smooth function such that $0\le g(x)\le 1$, $g(x)=0$ if $|x|\ge 1$ and $\int_{\R^2} g^2(x)dx=1$. For $s>0$ (small), denote $g_s(x)=s^{-1}g(x/s)$ and  
$$\Pi_{s,u,p}  = \left| {f_{s,u,p} } \right\rangle \left\langle {f_{s,u,p} } \right| ~~\text{where}~f_{s,u,p} (x) = e^{ip \cdot x} g_s(x - u)~\text{for all}~u,p~\in \R^2.
$$
\end{definition}
From the coherent identity,
\bq \label{eq:coherent-identity}
	(2\pi )^{ - 2} \iint\limits_{\mathbb{R}^2  \times \mathbb{R}^2 } {\Pi_{s,u,p}\,dpdu} = I~~\text{on}~L^2(\R^2),
\eq
it is straightforward to see that for any density matrix $\gamma$ and for any potential $V$ satisfying $V_+\in L^1(\R^2)$,
\bq
  \Tr \left[ { - h^2\Delta \gamma } \right] &=& (2\pi )^{ - 2} \iint {\Tr \left[ { -h^2\Delta \Pi_{s,u,p} \gamma } \right]}dpdu  \nn\hfill \\
   &=& (2\pi )^{ - 2} \iint {h^2p^2 \Tr \left[ {\Pi_{s,u,p} \gamma } \right]}dpdu - || \nabla g||_{L^2}^2 h^2s^{-2}\Tr (\gamma ) ,\label{eq:coherent-state-identity-1}\hfill\\
 \Tr [(-V*g_s^2 )\gamma ] &=& (2\pi )^{ - 2} \iint {\Tr [(-V*g^2 )\Pi_{s,u,p} \gamma ]}dpdu \nn\hfill\\
&=& (2\pi )^{ - 2} \iint {-V(u)\Tr [\Pi_{s,u,p} \gamma ]}dpdu.\label{eq:coherent-state-identity-2}
\eq
Motivated by \eqref{eq:coherent-state-identity-2}, it is useful to have some estimate for $(V-V*g_s^2)$. The proof of the following lemma can be found in the Appendix.

\begin{lemma}\label{le:V-Vg2} If $V$ is either $V^{\TF}_\lambda$ or $(|x|^{-1}-1)$ and $\Lambda=|\ln h|$, $r=h^{1/2}$, $s=h^{2/3}$ then  
\[
\int_{r \le |x| \le 2\Lambda } {|V - V*g_s^2 |^2(x)dx}  \le Ch^{1/4}.
\]
\end{lemma}

\begin{lemma}[Semiclassics in exterior region] \label{le:semiclassical-approximation-intermediate-region} Let $V$ be either $V^{\TF}_{\lambda}$ or $(|x|^{-1}-1)$. When $\Lambda=|\ln h|$ and $r=h^{1/2}\to 0$ one has
\bqq
	\Tr \left[ {\Phi_2 \left( { - h^2\Delta -V} \right)\Phi_2 } \right]_ -   =-(8\pi h^2)^{-1}\int {V_ + ^2(x)  \Phi_2 ^2 (x)dx}+ o(h^{-2}).
\eqq
 \end{lemma}
 
\begin{proof} 1. To prove the lower bound, we choose the density matrix
\[
\gamma_2 := 1_{( - \infty ,0]} \left[ {\Phi_2 \left( { -h^2\Delta  - V} \right)\Phi_2} \right].
\]
Taking $s=h^{2/3}$ and using identities \eqref{eq:coherent-state-identity-1} and \eqref{eq:coherent-state-identity-2} we can write   
\bq\label{eq:semiclassical-chir-lower-eq1}
	\Tr \left[ { \Phi_2 \left( { - h^2\Delta  - V} \right)\Phi_2 } \right]_ -  &= & \Tr \left[ {\left( { -h^2\Delta  - V} \right)\Phi_2 \gamma _2 \Phi_2} \right] \nn \hfill \\
	&=&  (2\pi )^{ - 2} \iint {\left[ {h^2 p^2  - V  (u)} \right]\Tr \left[ {\Pi _{s,u,p} \Phi_2 \gamma_2 \Phi_2} \right]}\,dpdu\nn \hfill\\
	&~&+\Tr \left[ { \left( { V  *g_s^2-V-C h^2s^{-2}} \right)\Phi_2 \gamma _2\Phi_2} \right] .
\eq

2. To bound the second term of the right-hand side of \eqref{eq:semiclassical-chir-lower-eq1}, we can apply H\"older's inequality, Lemma~\ref{le:hydrogen-2} (i) with $\Omega:=\supp \Phi_2 \subset \{r\le |x| \le 2 \Lambda\}$ and Lemma~\ref{le:V-Vg2}  to get
\bq \label{eq:semiclassics-exterior-second-lowerbound}
	&~&\Tr \left[ { \left( { V  *g_s^2-V-C h^2s^{-2}} \right)\Phi_2 \gamma _2\Phi_2} \right] \nn \\
	&\ge& - \left\| {V*g_s^2  - V} \right\|_{L^2 (\Omega )} \left\| {\rho _{\Phi _2 \gamma _2 \Phi _2 } } \right\|_{L^2 (\mathbb{R}^2 )}  - Ch^2 s^{ - 2} \operatorname{Tr} [\Phi _2 \gamma _2 \Phi _2 ] \nn\hfill \\
   &\geqslant&  - Ch^{ - 2} \left\| {V*g_s^2  - V} \right\|_{L^2 (\Omega )} \left\| {V_ +  } \right\|_{L^2 (\Omega )}  - Cs^{ - 2} \left\| {V_ +  } \right\|_{L^2 (\Omega )} |\Omega |^{1/2} \nn \hfill \\
   &\geqslant&  - Ch^{ - 2} h^{1/8} |\ln h|^{1/2}  - Cs^{ - 2} |\ln h|^{1/2} \Lambda  = o(h^{ - 2} ) .
\eq

For the first term of the right-hand side of  \eqref{eq:semiclassical-chir-lower-eq1}, because
$$
0 \le \Tr \left[ {\Pi _{s,u,p} \Phi_2\gamma _2\Phi_2 } \right]\le \Tr \left[ {\Pi _{s,u,p} \Phi_2^2 } \right] = (\Phi_2^2 *g_s^2 )(u)
$$
we obtain
\bq \label{eq:semiclassics-exterior-first-lowerbound}
	&~&(2\pi )^{ - 2} \iint {\left[ {h^2 p^2  - V  (u)} \right]\Tr \left[ {\Pi _{s,u,p} \Phi_2 \gamma_2 \Phi_2} \right]}\,dpdu\nn \hfill \\
	&\ge&  - (2\pi )^{ - 2} \iint {\left[ {h^2p^2  - V (u)} \right]_ -  (\Phi_2^2 *g_s^2 )(u)}dpdu \nn \hfill \\
	&=&  - (8\pi h^2)^{ - 1}  \int {V_ + ^2 (u)(\Phi_2^2 *g_s^2 )(u)du} \nn \hfill\\
	&=&  -(8\pi h^2 )^{ - 1}  \int {V_ + ^2 (u)\Phi_2^2 (u)du} + o(h^{ - 2} ). 
\eq
Here the last estimate follows from 
\bq \label{eq:upper-bound-chirV-V*g2}
\int {V_ + ^2 (u)|\Phi_2^2  - \Phi_2^2 *g_s^2 |(u)du}  \le Csr^{-1}\int_{|x|\ge r/2}V_+(u)^2du \le Csr^{-1}|\ln r| =o(h^{-2}),
\eq
where we have used $|\Phi_2^2  - \Phi_2^2 *g_s^2 |(x)\le Csr^{-1}1_{\{|u|\ge r/2\}}$ when $|x|\ge r \gg s$. Replacing (\ref{eq:semiclassics-exterior-second-lowerbound}) and (\ref{eq:semiclassics-exterior-first-lowerbound}) into \eqref{eq:semiclassical-chir-lower-eq1} we get the lower bound in the lemma. 

3. To show the upper bound, we choose 
\[
	\gamma ^{(2)} : = (2\pi )^{ - 2} \iint\limits_M {\Pi _{s,u,p}\,dpdu}, ~~M: = \left\{ {(u,p): h^2p^2  - V(u) \le  0} \right\}.
\]
Using the coherent identity \eqref{eq:coherent-identity} and the IMS formula, it is straightforward to compute that
\bq \label{eq:Hydrogen-comparison-exterior-upper-bound}
&~&\Tr \left[ {\Phi_2 \left( { - h^2 \Delta  - V} \right)\Phi_2 } \right]_ -   \le \Tr \left[ {\Phi_2 \left( { - h^2 \Delta  - V} \right)\Phi_2 \gamma ^{(2)}  } \right] \nn\hfill \\
   &=& (2\pi )^{ - 2} \iint\limits_M {\left\langle {f_{s,u,p} (x)} \right|\Phi_2 (x)\left( { - h^2 \Delta _x  - V(x)} \right)\Phi_2 (x)\left| {f_{s,u,p} (x)} \right\rangle _{L^2 (\mathbb{R}^2 ,dx)} dpdu} \nn\hfill \\
   &=& (2\pi )^{ - 2} \iint\limits_M {\left\langle {e^{ip \cdot x} } \right| - (g_s \Phi_2 )^2 \frac{{h^2 }}
{2}\Delta  - \frac{{h^2 }}
{2}\Delta (g_s \Phi_2 )^2  + h^2 |\nabla (g_s \Phi_2 )|^2 - (g_s\Phi_2)^2 V\left| {e^{ip \cdot x} } \right\rangle dpdu} \nn\hfill \\
   &=& (2\pi )^{ - 2} \iint\limits_M {\left( {h^2 p^2 (\Phi_2^2 *g_s ^2 )(u) + h^2 \int\limits_{\mathbb{R}^2 } {|\nabla (g_s \Phi_2)(x)|^2 dx}  - ((\Phi_2^2 V)*g_s ^2 )(u)} \right)dpdu} \nn\hfill \\
   &=&  - (8\pi h^2 )^{ - 1} \int {V_ + ^2 (u)\Phi_2^2 (u)du}  + (8\pi h^2 )^{ - 1} \int {\left[ {V_ + ^2 (\Phi_2^2 *g_s ^2 ) -V_ + ^2 \Phi_2^2} \right]du}  \nn\hfill \\
&~&+(4\pi )^{ - 1} \iint {V_+(u)|\nabla (g_s \Phi_2 )(x)|^2 dxdu}  +(4\pi h^2 )^{ - 1}  \int {\Phi_2^2 V\left[ {V_ +   - (V_ +  *g_s ^2 )} \right]du} .
 \eq
 
4. Finally we verify that the last three terms of the right-hand side of \eqref{eq:Hydrogen-comparison-exterior-upper-bound} are of $o(h^{-2})$. The second term was already treated by \eqref{eq:upper-bound-chirV-V*g2}. Using
\[
\int_{\R^2}|\nabla (g_s \Phi_2 )(x)|^2 dx \le C(r^{ - 2}  + s^{ - 2} )1_{\{r/2\le |u| \le 3 \Lambda\}}.
\]
we can bound the third term as
\[
\iint {V_ +  (u)|\nabla (g_s \Phi _2 )(x)|^2 dxdu} \leqslant C(r^{ - 2}  + s^{ - 2} )\int_{r/2 \leqslant |x| \leqslant 3\Lambda } {V_ +  (u)du}  = o(h^{ - 2} ).
\]

To estimate the last term, we introduce a universal constant $\Lambda_0>0$ such that $V(x)\ge 0$ when $|x|\le 2\Lambda_0$ (such $\Lambda_0$ exists since $|V(x)-|x|^{-1}|\le C(|x|^{-1/2}+1)$). Using  $V_+(u)=V(u)$ and $(V_+*g_s^2)(u)=(V*g_s^2)(u)$ when $|u|\le \Lambda_0$, and Lemma~\ref{le:V-Vg2} we get
\[
\int_{|u| \leqslant \Lambda _0 } {|\Phi _2^2 V|.|V_ +   - V_ +  *g_s^2 |du}  \leqslant \left\| V \right\|_{L^2 (\Omega ) } \left\| {V*g_s^2  - V} \right\|_{L^2 (\Omega) } \le C|\ln h|^{1/2}h^{1/8} = o(1)
\]
where $\Omega =\supp \Phi_2 \subset \{r\le |u| \le \Lambda\}$. On the other hand, because $|V(u)|\le C$ when $|u|\ge \Lambda_0$ and $V_+\in L^1(\R^2)$, 
\[
\int_{|u| \geqslant \Lambda _0 } {|\Phi _2^2 V|.|V_ +   - V_ +  *g_s^2 |du}  \leqslant C\left\| {V_ +  *g_s^2  - V_ +  } \right\|_{L^1 (\mathbb{R}^2 )}  = o(1).
\]
Thus the last term of the right-hand side of \eqref{eq:Hydrogen-comparison-exterior-upper-bound} is also of $o(h^{-2})$. This completes the proof.
\end{proof}

Lemmas \ref{le:localization}, \ref{le:Hydrogen-comparison-interior-region} and \ref{le:semiclassical-approximation-intermediate-region} together yield \eqref{eq:hydrogen-comparison}, which is equivalent to \eqref{eq:semiclassical-approximation}. 

\subsection{Trial Density Matrix}
The last step in proving Theorem~\ref{thm:semiclassical-approximation} is to construct a trial density matrix. 

\begin{lemma}
There exists a density matrix $\gamma_h$ satisfying \eqref{eq:trial-density-ineq1} and \eqref{eq:trial-density-ineq2}. 
\end{lemma}
\begin{proof}
Recall that we always choose $\Lambda=|\ln h|$, $r=h^{1/2}$ and $s=h^{2/3}$.

1. From the proof of Lemmas \ref{le:localization}, \ref{le:Hydrogen-comparison-interior-region} and \ref{le:semiclassical-approximation-intermediate-region}, if we choose the density matrices 
\bqq
\gamma ^{(1)} &:=& 1_{( - \infty ,0]} \left[ {\Phi_1 \left( { - h^2 \Delta  - |x|^{-1}+1} \right) \Phi_1} \right] ,\hfill\\
\gamma ^{(2)} &:=& (2\pi )^{ - 2} \iint\limits_{h^2p^2  - V_\lambda^{\TF}(u) \le  0} {\Pi _{s,u,p}\,dpdu},\hfill\\
\gamma^{(0)}&:=&\Phi_1 \gamma^{(1)} \Phi_1 + \Phi_2 \gamma^{(2)}\Phi_2 
\eqq
then 
\bq \label{eq:pre-trial-1}
\Tr[-h^2\Delta -V^{\TF}_\lambda]_-=\Tr[(-h^2\Delta -V^{\TF}_\lambda)\gamma^{(0)}]+o(h^{-2}).
\eq

2. Using the coherent identity \eqref{eq:coherent-identity} and the TF equation $\rho^{\TF}_\lambda=(2\pi)^{-1}[V^{\TF}_\lambda]_+$, we can compute explicitly that
\bqq
	\rho _{\gamma ^{(2)}} (x) &:=& \gamma ^{(2)} (x,x) = (2\pi )^{ - 2} \iint\limits_{h^2 p^2  - V_\lambda ^{\TF } (u) \le 0} {\Pi _{s,u,p} (x,x)\, dpdu} \hfill \\
	&=& (4\pi h^2 )^{ - 1} ([V_\lambda ^{\TF } ]_ +  *g_s^2 )(x) = (2h^2 )^{ - 1} (\rho _\lambda ^{\TF } *g_s^2 )(x) .
 \eqq
Therefore,
\bq \label{eq:pre-trial-density}
2h^2 \rho _{\gamma ^{(0)} } = 2h^2\rho _{\Phi _1\gamma ^{(1)} \Phi _1}+\Phi_2^2\rho _\lambda ^{\TF } *g_s^2.
\eq
Since $\int \rho _\lambda ^{\TF } *g_s^2 =\int \rho_\lambda^{\TF}$ and $\Tr[\Phi _1\gamma ^{(1)} \Phi _1] \le Crh^{-2}|\ln h|^{1/2}$ (see \eqref{eq:Tr-Phi1-gamma1-Phi1}), we have 
\bq \label{eq:pre-trial-2}
2h^2\int_{\R^2} \rho _{\Phi _1\gamma ^{(1)} \Phi _1}(x)dx \le \int_{\R^2}\rho_\lambda^{\TF}(x)dx+ Cr|\ln h|^{1/2}.
\eq
On the other hand, we can write from \eqref{eq:pre-trial-density} that 
$$
2h^2 \rho _{\gamma ^{(0)} }-\rho_\lambda^{\TF} = 2h^2\rho _{\Phi _1\gamma ^{(1)} \Phi _1}+\Phi_2^2 (\rho _\lambda ^{\TF } *g_s^2-\rho_\lambda^{\TF}) + (1-\Phi_2^2) \rho_\lambda^{\TF}.
$$
Since $\rho _\lambda ^{\TF}\in L^{4/3}(\R^2)$, we have $\rho _\lambda ^{\TF } *g_s^2-\rho _\lambda^{\TF}$ and $ (\Phi_2^2-1) \rho _\lambda^{\TF}$ converge to $0$ in  $L^{4/3}(\R^2)$. Moreover, using Lemma~\ref{le:hydrogen-2} we have $2h^2\rho _{\Phi _1\gamma ^{(1)} \Phi _1} \to 0$ in $L^{4/3}(\R^2)$. Thus $2h^2 \rho _{\gamma ^{(0)} } - \rho _\lambda^{\TF}\to 0~\text{in}~L^{4/3}(\R^2).$ Since the Coulomb norm is dominated by the $L^{4/3}$-norm (see Theorem~\ref{thm:Coulomb-norm}), we then also have 
\bq \label{eq:pre-trial-3}
D(2h^2 \rho _{\gamma ^{(0)} } - \rho _\lambda^{\TF}) \to 0.
\eq

3. Finally, we choose $\ell$ such that $|\ln h|^{-1}\gg \ell \gg r |\ln h|^{1/2}$ (e.g. $\ell=r^{1/2} = h^{1/4}$) and define
$$\gamma_h:=(1- \ell )\gamma^{(0)}.$$
Then using \eqref{eq:pre-trial-2} and $\ell \gg r |\ln h|^{1/2}$ we have
$$ 2h^2\Tr(\gamma_h) \le (1-\ell)(1+Cr|\ln h|^{1/2}) \int \rho^{\TF}_\lambda \le \int \rho^{\TF}_\lambda$$
for $h$ small enough. Moreover, since $|\ln h|^{-1}\gg \ell$, the inequalities \eqref{eq:pre-trial-1} and \eqref{eq:pre-trial-3} still hold true with $\gamma^{(0)}$ replaced by $\gamma_h$.
\end{proof}

\section{Proofs of the Main Theorems}\label{sec:results}
\subsection{Ground State Energy}\label{sec:GSE}
Having the semiclassics in Theorem~\ref{thm:semiclassical-approximation}, the proof of Theorem~\ref{thm:GSE} is standard (see \cite{Li81}).

\begin{proof}[Proof of Theorem~\ref{thm:GSE}] 1. We first prove the lower bound. Taking any (normalized) wave function $\Psi\in \bigwedge_{i = 1}^N L^2 (\mathbb{R}^2)$, we need to show that 
$$\left({\Psi ,H_{N,Z} \Psi } \right) \ge -\frac{1}{2}Z^2\ln Z+E^{\TF}(\lambda)Z^2+o(Z^2).$$
Starting with the Lieb-Oxford inequality (Theorem~\ref{thm:Lieb-Oxford-inequality})
\[
	\left( {\Psi , \sum\limits_{1 \le i < j \le N} {\frac{1}{{|x_i  - x_j |}}} \Psi } \right)  \ge D(\rho _{\Psi  } ) - C_{\mathrm{LO}}
	\int {\rho _{\Psi  } ^{3/2} }, 
\]
we want to bound $\int {\rho _{\Psi  } ^{3/2}}$. It of course suffices to assume that $\left( {\Psi ,H_{N,Z} \Psi } \right) \le 0$. 
Using the Lieb-Thirring inequality \eqref{eq:LT-kinetic}, the hydrogen spectrum in Theorem~\ref{thm:hydrogen-spectrum-original} and $\Tr(\gamma_\Psi) = N \le CZ$ we arrive at
\bqq
	0 \ge 4 \left( {\Psi, H_{N,Z} \Psi } \right)  &\ge& \Tr [-\Delta \gamma_\Psi] +\Tr \left[ {\left( { - \Delta  - 4Z|x| ^{-1}} \right)\gamma _\Psi  } 
	\right] \hfill\\
	&\ge& K_2 \int_{\mathbb{R}^2 } {\rho _{\Psi  }^2 (x)dx}-CZ^2 |\ln Z|
\eqq
By H\"older's inequality and $\int \rho_\Psi = N\le CZ$ again we conclude 
$$ \int_{\mathbb{R}^2 } {\rho _{\Psi  }^{3/2} (x)dx}  \le \left( {\int_{\mathbb{R}^2 } {\rho _{\Psi  }^2 (x)dx} } \right)^{1/2} \left( {\int_{\mathbb{R}^2 } 
{\rho _{\Psi  }^{} (x)dx} } \right)^{1/2} \le CZ^{3/2} |\ln Z|^{1/2}.$$
Thus the Lieb-Oxford inequality gives 
\bq \label{eq:GSE-upper-0}
	\left( {\Psi, H_{N,Z} \Psi } \right)  &\ge& \Tr \left[ {\left( { - \frac{1}{2}\Delta  - Z|x|^{ - 1} } \right)\gamma _\Psi  } \right] + D(\rho _{\gamma _\Psi  } ) - 
	CZ^{3/2} |\ln Z|^{1/2} \nn  \hfill \\
	&=& Z\Tr \left[ {\left( { - (2Z)^{ - 1} \Delta  - V_\lambda ^{\TF }   } \right)\gamma _\Psi  } \right] - Z^2\left[ {\mu _\lambda ^{\TF } (N/Z) + D(\rho _
	\lambda ^{\TF } )} \right] \nn \hfill\\
	&~&+ Z^2 D(Z^{ - 1} \rho _{\gamma _\Psi  }  - \rho _\lambda ^{\TF } ) - CZ^{3/2} |\ln Z|^{1/2}.
\eq
For the lower bound, we can ignore the nonnegative term $D(Z^{ - 1} \rho _{\gamma _\Psi  }  - \rho _\lambda ^{\TF } ) \ge 0$. With the semiclassics of the TF potential in Theorem~\ref{thm:semiclassical-approximation} and $h^2=(2Z)^{-1}$, one has
\[
\begin{gathered}
~~~  \Tr \left[ {\left( { - (2Z)^{ - 1} \Delta  - V_\lambda ^{\TF } } \right)\gamma _\Psi  } \right] \ge \Tr \left[ { - (2Z)^{ - 1} \Delta  - V_\lambda ^{\TF }   } \right]_ -   \hfill \\
   \ge  - \frac{1}
{2}Z\ln Z + Z\left[ { - (4\pi )^{ - 1} \int {\left( {[V_\lambda ^{\TF } (x)]_ + ^2  - [|x|^{ - 1}  - 1]_ + ^2 } \right)dx}  + \frac{1}{2}c^{\mathrm{H} } } \right] + o(Z). \hfill \\ 
\end{gathered} 
\]
Together with $N/Z\to \lambda$, we obtain from \eqref{eq:GSE-upper-0} that
\[
\left( {\Psi, H_{N,Z} \Psi } \right)  \ge  - \frac{1}{2}Z^2 \ln Z + e(\lambda )Z^2  + o(Z^2 )
\]
where
\[
e (\lambda ): =  - (4\pi )^{ - 1} \int {\left( {[V_\lambda ^{\TF } (x)]_ + ^2  - [|x|^{ - 1}  - 1]_+ ^2 } \right) dx}  - \mu _\lambda ^{\TF } \lambda  - D(\rho _\lambda ^{\TF } ) + \frac{1}{2}c^{\mathrm{H} } .
\]
By the TF equation $2\pi \rho^{\TF}_\lambda = [V^{\TF}_\lambda]_+$ we have
\[
 - [V_\lambda ^{\TF } ]_ + ^2  = [V_\lambda ^{\TF } ]_ + ^2  - 2[V_\lambda ^{\TF } ]_ + ^{} V_\lambda ^{\TF }  = 4\pi ^2 [\rho _\lambda ^{\TF } ]^2  - 4\pi \rho _\lambda ^{\TF } [|x|^{ - 1}  - \rho _\lambda ^{\TF } *\clmb^{ - 1}  - \mu _\lambda ^{\TF } ].
\]
Replacing this identity and $\mu _\lambda ^{\TF } \lambda = \mu _\lambda ^{\TF } \int {\rho _\lambda ^{\TF } }$ into the definition of $e(\lambda)$, we see that $e(\lambda)=E^{\TF}(\lambda)+ c^{\rm H}/2$. Thus we get the lower bound on the ground state energy.

2. To show the upper bound, because $\lambda\mapsto E^{\TF}(\lambda)$ is continuous, it suffices to show that for any $0<\lambda'<\lambda$ fixed, one has
\bqq \label{eq:GSE-upper-condtion-0}
	E(N,Z) \le  - \frac{1}{2}Z^2 \ln Z + E^{\TF } (\lambda ')Z^2  + 
	o(Z^2 ).
\eqq
Using Lieb's variational principle (see Theorem~\ref{thm:Lieb-variational-principle}) we want to find a density matrix $\gamma$ such that $\Tr(\gamma)\le N$ and
\bqq \label{eq:GSE-upper-condtion-0}
	\Tr \left[ {\left( { - \frac{1}{2}\Delta  - Z|x|^{ - 1} } \right)\gamma } \right] + D(\rho _\gamma  ) \le  - \frac{1}{2}Z^2 \ln Z + E^{\TF } (\lambda ')Z^2  + 
	o(Z^2 ).
\eqq
This condition can be rewritten, using the same calculation of proving the lower bound (see \eqref{eq:GSE-upper-0}), as
\bq \label{eq:GSE-upper-condtion-1}
\Tr \left[ {\left( { - (2Z)^{ - 1} \Delta  - V_{\lambda '}^{\TF } } \right)\gamma } \right] + ZD(Z^{ - 1} \rho _\gamma   - \rho _{\lambda '} ^{\TF } ) \le \Tr \left[ { - (2Z)^{ - 1} \Delta  - V_{\lambda '} ^{\TF } } \right]_ -   + o(Z).
\eq
According to Theorem~\ref{thm:semiclassical-approximation} with $h^2=(2Z)^{-1}$, we can find a trial density matrix $\gamma$ satisfying (\ref{eq:GSE-upper-condtion-1}) such that $\Tr(\gamma)\le Z\int \rho^{\TF}_{\lambda '} \le \lambda ' Z$. Since $N/Z\to \lambda >\lambda '$, one has $\Tr(\gamma)\le \lambda ' Z \le N$ for $Z$ large enough and it ends the proof.
\end{proof}

\subsection{Extensivity of Neutral Atoms}
\begin{proof}[Proof of Theorem~\ref{thm:radius-atom}]
Let $\theta _R$ be a smooth function such that $\theta _R(x)=0$ if $|x|\le R$ and $\theta _R(x)=1$ if $|x| \ge 2R$. From the proof of Theorem~\ref{thm:semiclassical-approximation} and Theorem~\ref{thm:GSE}, we have, with $\gamma:=\gamma_{\Psi_{N,Z}}$ and $h^2=(2Z)^{-1}$, 
\[
\Tr [( - h^2 \Delta  - V_1 ^{\TF } )\gamma ] = \Tr [ - h^2 \Delta  - V_1 ^{\TF } ]_ -   + o(h^{ - 2} ).
\]
Using the localization as in Lemma~\ref{le:localization} and the semiclassics of Lemma~\ref{le:semiclassical-approximation-intermediate-region} we get  
\bq \label{eq:radius-bound-1}
	\Tr [\theta _R ( - h^2 \Delta  - V_1 ^{\TF } )\theta _R \gamma ] &\le& \Tr [\theta _R ( - h^2 \Delta  - V_1 ^{\TF }  )\theta _R ]_ -   
	+ o(h^{ - 2} ) \nn \hfill\\
	&=& -(8\pi h^2)^{ - 1} \int {[V_1^{\TF } ]_ + ^2 (x)\theta _R^2 (x)dx}+o(h^{-2}).
\eq
On the other hand, since $V_1^{\TF} \le |x|^{-1}\le R^{-1}$ in $\supp \theta_R$,
\bq\label{eq:radius-bound-3}
	\Tr [\theta _R ( - h^2 \Delta  - V_1^{\TF}  )\theta _R \gamma ] \ge  -R^{-1}\Tr [\theta _R \gamma \theta _R] =-R^{-1}\int {\theta _R^2 (x)\rho _
	\gamma  (x)dx}.
\eq
Putting \eqref{eq:radius-bound-1} and \eqref{eq:radius-bound-3} together we arrive at
\[
\int {\theta _R^2 (x)\rho _\gamma  (x)dx}  \ge R(8\pi h^2 )^{ - 1} \int {[V_1^{\TF } ]_ + ^2 (x)\theta _R^2 (x)dx}  + o(h^{ - 2} ).
\]
Replacing $h^2=(2Z)^{-1}$, we can conclude that  
\bqq
	\int_{|x| \ge R} {\rho _\gamma  (x)dx} \ge \int {\theta _R^2 (x)\rho _\gamma  (x)dx} \ge C_RZ + o(Z)
\eqq
where
\bqq
	C_R:= R(4\pi )^{ - 1} \int {[V_1^{\TF} ]_ + ^2 (x)\theta _R^2 (x)dx} \ge \pi R\int_{|x| \ge 2R} {(\rho _1^{\TF } (x))^2 dx} .
\eqq
Note that $C_R>0$ because $\supp \rho^{\TF}_1$ is unbounded (see Theorem~\ref{thm:TF-theory}).
\end{proof}

\subsection{Semiclassics for Coulomb Singular Potentials}\label{app:general-semiclassics}
\begin{proof}[Proof of Theorem~\ref{thm:semiclassics-general-potential}]  We shall show how to adapt the proof of \eqref{eq:semiclassical-approximation} in the previous section to the general case. We however leave some details to the readers. By scaling we can assume $\kappa=1$.

1. The main difficulty of the general case is that we do not have the estimate in Lemma~\ref{le:V-Vg2} in the exterior region. Therefore, we need a more complicated localization. Let $r=h^{1/2},s=h^{2/3}$ and let $g_s$ be as in Definition~\ref{def:coherent}. For any $\eps>0$ small, denote  
\[
W(\eps ,h): = \int_{\eps  \le |x| \le \eps ^{ - 1} } {|V|^2 dx} \int_{\eps  \le |x| \le \eps ^{ - 1} } {|V_ +   - V_ +  *g_s^2 |^2 dx} 
\]
Because $V\in L^2_{\loc}(\R^2\minus \{0\})$, for any $\eps>0$ fixed we have $W(\eps,h)\to 0$ as $h\to 0^+$. Therefore, we can choose $\eps=\eps(h)$ such that $\eps(h)\ge |\ln h|^{-1}$, $\eps(h) \to 0$ and $W(\eps(h),h)\to 0$ as $h\to 0^+$. Let $\varphi$ as in Definition~\ref{def:partition-unity} and define 
\bqq
\widetilde\Phi _1 (x)&=&\varphi(x/r), \\
 \widetilde\Phi _2(x)&=& (1-\varphi^2(x/r))^{1/2}\varphi (x/\eps ) ,\\
\widetilde \Phi_3(x) &=& (1-\varphi^2(x/\eps)^{1/2}\varphi(x \eps/2) ,\\
\widetilde \Phi_4(x) &=& (1-\varphi^2(x \eps/2))^{1/2}.
\eqq
Then $\sum_{i=1}^4\widetilde\Phi_i^2=1$, $\supp \widetilde\Phi_1\subset \{|x|\le 2r\}$, $\supp \widetilde\Phi_2 \subset \{r\le |x|\le 2\eps  \}$, $\supp \widetilde\Phi_3 \subset \{\eps \le |x| \le \eps^{-1}\}$, and $\supp \widetilde\Phi_4\subset \{|x| \ge (2\eps)^{-1}\}$.

2. Following the proof of Lemma~\ref{le:localization} we can show that   
\bq 
\Tr [ - h^2 \Delta  - V ]_ - = \sum_{i=1}^3 \Tr [\widetilde\Phi_i ( - h^2 \Delta  - V)\widetilde\Phi_i ]_ -+o(h^{-2}).\label{eq:general-localization}
\eq
Note that the assumptions $1_{\{|x|\ge 1\}}V_+\in L^2(\R^2)$ and $|\ln h| \ge \eps^{-1}\to \infty$ is sufficient to bound the contribution of the region $\supp \widetilde\Phi_4$ by the Lieb-Thirring inequality \eqref{eq:LT-eigenvalue-sum}. To control the localization cost in the region $\supp \widetilde\Phi_3$, we may use Lemma~\ref{le:hydrogen-2} (i) instead of Lemma~\ref{le:hydrogen-2} (ii). 

3. Because  $|V(x)-|x|^{-1}+1|\le C(|x|^{-\theta}+1)\le Cr^{-\theta}$ for $x \in \supp \widetilde\Phi_1$, we can follow the proof of Lemma ~\ref{le:Hydrogen-comparison-interior-region} to get
\bq
\Tr \left[ {\widetilde\Phi _1 \left( { - h^2\Delta -V} \right)\widetilde\Phi _1 } \right]_ -=\Tr \left[ {\widetilde\Phi _1 \left( { - h^2\Delta  - |x|^{-1} + 1} \right)\widetilde\Phi _1 } \right]_ - + o(h^{-2}) \label{eq:general-interior}.
\eq

4. Adapting the coherent state approach in the proof of Lemma~\ref{le:semiclassical-approximation-intermediate-region}, we can show that
\bq 
\Tr \left[ {\widetilde\Phi_3 \left( { - h^2\Delta -V} \right)\widetilde\Phi_3 } \right]_ -  =-(8\pi h^2)^{-1}\int {V_ + ^2(x)  \widetilde\Phi_3^2 (x)dx}+ o(h^{-2}) \label{eq:general-exterior}.
\eq
To obtain the lower bound it suffices to consider $\Tr [ {\widetilde\Phi_3 \left( { - h^2\Delta -V_+} \right)\widetilde\Phi_3 } ]_ - $ and then use the assumption $W(\eps(h),h)\to 0$ instead of Lemma~\ref{le:V-Vg2} in \eqref{eq:semiclassics-exterior-second-lowerbound}. When proving the upper bound, the assumption $W(\eps(h),h)\to 0$ is again enough to estimate the last term of (\ref{eq:Hydrogen-comparison-exterior-upper-bound}).

5. In the intermediate region $\supp \widetilde\Phi_2 \subset \{r\le |x|\le 2\eps \}$, we have 
$$V_1(x):=|x|^{-1}+C|x|^{-\theta} \ge V(x) \ge |x|^{-1}-C|x|^{-\theta}=:V_2(x) \ge 0.$$
We start with the lower bound
$$ \Tr \left[ {\widetilde\Phi_2 \left( { - h^2\Delta -V} \right)\widetilde\Phi_2 } \right]_ - \ge \Tr \left[ {\widetilde\Phi_2 \left( { - h^2\Delta -V_1} \right)\widetilde\Phi_2} \right]_ - .$$
Using the coherent state approach as in the proof of Lemma~\ref{le:semiclassical-approximation-intermediate-region}, we can show that
\bqq 
	\Tr \left[ {\widetilde\Phi_2 \left( { - h^2\Delta -V_1} \right)\widetilde\Phi_2} \right]_ - =\int {V_1}^2(x)\widetilde\Phi_2(x)^2dx +o(h^{-2}) .
\eqq
To do that, we just need to replace Lemma~\ref{le:V-Vg2} by the following estimate
\[
	\int_{r \le |x| \le 2\eps } {|V_{1}  - V_{1} *g_s^2 |^2 dx}  \le Cs^2 |\ln r|.
\]
Moreover, since $\supp \widetilde\Phi_2 \subset \{r\le |x|\le 2\eps \}$ with $\eps=\eps(h)\to 0$ and $|V-V_1| \le C|x|^{-\theta}$, we have
$$ \int V_1^2(x)\widetilde\Phi_2(x)^2dx = \int {V}^2(x)\widetilde\Phi_2(x)^2dx +o(h^{-2}).$$
Therefore, we arrive at
$$ \Tr \left[ {\widetilde\Phi_2 \left( { - h^2\Delta -V} \right)\widetilde\Phi_2 } \right]_- \ge \int {V}_+^2(x)\widetilde\Phi_2(x)^2dx +o(h^{-2}).$$
Similarly, again using the coherent state approach we get the reverse inequality
\bqq &~&\Tr \left[ {\widetilde\Phi_2 \left( { - h^2\Delta -V} \right)\widetilde\Phi_2 } \right]_-  \le \Tr \left[ {\widetilde\Phi_2 \left( { - h^2\Delta -V_2} \right)\widetilde\Phi_2 } \right]_- \\
&=&  \int {V_2}^2(x)\widetilde\Phi_2(x)^2dx +o(h^{-2}) =\int {V}^2(x)\widetilde\Phi_2(x)^2dx +o(h^{-2})  .
\eqq
Thus we obtain the semiclassics
\bq 
	\Tr \left[ {\widetilde\Phi_2 \left( { - h^2\Delta -V} \right)\widetilde\Phi_2 } \right]_ -=\int {V}_+^2(x)\widetilde\Phi_2(x)^2dx +o(h^{-2}).
	\label{eq:general-intermediate}
\eq

6. The desired semiclassics follows from \eqref{eq:general-localization}, \eqref{eq:general-interior}, \eqref{eq:general-exterior} and \eqref{eq:general-intermediate}.
\end{proof}

\begin{appendix}
\section{Appendix}

In this appendix we provide several technical proofs.

\begin{proof}[Proof of Theorem~\ref{thm:HVZ}]
(i) The HVZ Theorem indeed holds for all dimension $d\ge 2$ (see e.g. \cite{Ma10} Theorem 2.1 for a short proof). The decay property is essentially taken from \cite{OO77} where the only change is of solving equation (3.8) in \cite{OO77}. In fact, the two-dimensional solution $w_2(r)$ (with $r=|x|$) is obtained by scaling the three-dimensional solution $w_3(r)$ in \cite{OO77} as $w_2=w_3|_{\eps\mapsto 4 \eps, Z\mapsto 4Z, r\mapsto r/2.}$

(ii) The proof of Zhislin's Theorem is standard and there is no difference between two and three dimensions. The  idea is that by induction we can use the ground state $H_{N,Z}$ to construct a $(N+1)$-particle wave function with strictly lower energy whereas $N<Z$. It should be mentioned that some certain decay of the ground state is necessary to control the localization error when we consider the cut-off wave function in a compact set. 

(iii) The asymptotic neutrality follows from the original proof in three dimensions of Lieb, Sigal, Simon and Thirring \cite{LSST88}. The key point of their proof is the construction of a partition of unity. But a partition of unity in three dimensions obviously yields a partition of unity in two dimensions, hence this part of the proof can be adopted. Note that the Pauli exclusion principle enters when solving the hydrogen atom.
\end{proof}

\begin{proof}[Proof of Lemma~\ref{le:hydrogen-1}]
For any $m\in \mathbb{N}$, one has
\bqq
	\Tr\left[ { - \frac{1}{2}\Delta  - |x|^{ - 1}  + \frac{1}{{2(m + 1/2)^2 }}} \right]_ -   &=& \sum\limits_{n = 0}^m {(2n + 1)\left[ { - \frac{1}{{2(n + 1/2)^2 }} + 
	\frac{1}{{2(m + 1/2)^2 }}} \right]} \hfill\\
	&=&- \sum\limits_{n = 0}^m {\frac{1}{{n + 1/2}}} +\frac{(m+1)^2}{2(m+1/2)^2}.
\eqq
Using Euler's approximation 
\[
	\sum\limits_{n = 0}^m {\frac{1}{{n + 1/2}}}  = \ln (m) + 2\ln (2) +\gamma_{E} +o(1)_{m\to \infty}
\]
we get
$$
   \Tr\left[ { - \frac{1}{2}\Delta  - |x|^{ - 1}  + \frac{1}{{2(m + 1/2)^2 }}} \right]_ -   = - \ln (m) - 2\ln (2) - \gamma_{E}  + \frac{1}{2} +o(1)_{m\to \infty}
$$
which implies \eqref{eq:asymtotic-hydrogen-mu}. Moreover, \eqref{eq:asymtotic-hydrogen-h} follows from \eqref{eq:asymtotic-hydrogen-mu} by scaling $x\mapsto (2h^{2})^{-1}x$,  namely
\bqq
  	\Tr\left[ { - h^2\Delta  - |x|^{ - 1}  + \mu} \right]_ -   &=& (2h^2)^{ -1} \Tr \left[ { - \frac{1}{2}\Delta  - |x|^{ - 1}  + 2h^2 \mu} \right]_ -.
\eqq
\end{proof}

\begin{proof}[Proof of Lemma~\ref{le:hydrogen-2}] (i) For any constant $a\ge 0$, using the Lieb-Thirring inequality \eqref{eq:LT-eigenvalue-sum} we have
\bqq
  0 \ge \Tr [( - h^2 \Delta  -V)\phi\gamma \phi] &\ge& \Tr [(-(h^2/2)\Delta +a) \phi\gamma \phi] + \Tr [( - (h^2/2) \Delta  -(a+V_+).1_{\Omega})\phi\gamma \phi] \\
  &\ge & \Tr [(-(h^2/2)\Delta +a) \phi\gamma \phi] - 4L_{1,2}h^{-4}|| a+V_+||_{L^2(\Omega)}^2.
\eqq
Choosing $a=1$ and using $\Tr [-\Delta \phi\gamma \phi] \ge 0$ and $(1+V_+)\in L^2(\Omega)$ we get $\Tr[\phi\gamma \phi]<\infty$, namely $\phi\gamma \phi$ is trace class. On the other hand, choosing $a=0$ and using  the Lieb-Thirring inequality \eqref{eq:LT-kinetic} to estimate $\Tr [-\Delta \phi\gamma \phi]$, we arrive at 
$$\int_{\R^2}\rho_{\phi\gamma\phi}^{2}(x)dx \le C h^{-4} ||V_+||_{L^2(\Omega)}^{2}.$$
Because $\supp \rho_{\phi\gamma\phi} \subset \Omega$, the above estimate and H\"older's inequality yield the desired bound on $\int \rho_{\phi \gamma \phi}^{2\alpha}$ for any $\alpha\in [0,1]$. 

(ii) We can use the same idea of the above proof. The only adaption we need in this case is to use both of the Lieb-Thirring inequality \eqref{eq:LT-eigenvalue-sum} and the hydrogen semiclassics \eqref{eq:asymtotic-hydrogen-h} to bound $\Tr [( - (h^2/2) \Delta  -(a+V_+).1_{\Omega})\phi\gamma \phi]$. More precisely, since $V\le C_0(|x|^{-1}+1)$, we have 
\bqq
  \operatorname{Tr} [( - (h^2/2) \Delta  - (a + V_ +  ).1_\Omega  )\phi \gamma \phi ] &\ge& \operatorname{Tr} [( - (h^2/4) \Delta  - C_0 |x|^{ - 1}  + 1)\phi \gamma \phi ] \\
&~&+ \operatorname{Tr} [( - (h^2/4) \Delta  - (C_0 +a+1).1_\Omega  )\phi \gamma \phi ] \\
&\ge & - Ch^{ - 2} |\ln h| - C(a+1)^2h^{ - 2} |\Omega | .
\eqq
\end{proof}

\begin{proof}[Proof of Lemma~\ref{le:convolution-bound}] 1. The lower bound follows from the radial symmetry of $\rho$ and the fact that $\Delta(|x|^{-1})=|x|^{-3}>0$ pointwise for all $x\ne 0$.

In fact, since $\rho$ is radially symmetric we can write
\[
(\rho *\clmb^{ - 1} )(x) = \int\limits_{|y| < |x|} {\left( {\int\limits_{S_y } {\frac{1}
{{|x - z_2 |}}dz_2 } } \right)\rho (y)dy}  + \int\limits_{|y| > |x|} {\left( {\int\limits_{S_x } {\frac{1}
{{|z_1  - y|}}dz_1 } } \right)\rho (y)dy} 
\]
where $dz_1$ and $dz_2$ are normalized Lebesgue measure on the circles $S_x:=\{z\in \R^2: |z|=|x|\}$ and $S_y:=\{z\in \R^2: |z|=|y|\}$. 

If $|x|>|y|$ then using the subharmonic property of the mapping $z\mapsto |x-z|^{-1}$ in the open set $\{z\in \R^2: |z|<|x|\}$ we get 
\[
\int\limits_{S_y } {\frac{1}
{{|x - z_2 |}}dz_2 }  \ge \frac{1}
{{|x|}}.
\]
Together with the similar inequality for $|y|>|x|$, we obtain the desired lower bound on $\rho*\clmb^{-1}$.

2. Because $\rho(x) \le (2\pi |x|)^{-1}$ and $\int \rho =\lambda$, for any $\kappa>1$, 
\bqq
	(\rho  *\clmb^{ - 1} )(x) &=& \int\limits_{\mathbb{R}^2 } {\frac{{\rho (y)}}{{|x - y|}}dy}  \le \int\limits_{|x - y| \le |x|/2} {}  + 
	\int\limits_{|x - y| \ge |x|/2,|y| \le \kappa |x|} {}  + \int\limits_{|y| \ge \kappa |x|} {}  \hfill \\
	& \le & \int\limits_{|x - y| \le |x|/2} {\frac{(2\pi)^{-1}}{{(|x|/2)|x - y|}}dy}  + \int\limits_{|y| \le \kappa |x|} {\frac{(2\pi)^{-1}}{{|y|(|x|/2)}}dy}  
	+ \int\limits_{\mathbb{R}^2 } {\frac{{\rho (y)}}{{(\kappa  - 1)|x|}}dy}  \hfill \\
	&\le & 1+ 2\kappa + \frac{\lambda}{{(\kappa  - 1)|x|}} .
\eqq
Optimizing the latter estimate over $\kappa>1$ yields the first upper bound on $\rho * \clmb^{-1}$. 

3. We now prove the second upper bound on $(\rho*\clmb^{-1})(x)$ for $|x|$ large. We start by decomposing $\mathbb{R}^2$ into three subsets  
\bqq
  \Omega_1  &:=& \left\{ {y \in \mathbb{R}^2 :|x - y| \ge |x|/2} \right\}, \hfill \\
  \Omega_2  &:=&\left\{ {y \in \mathbb{R}^2 : \left| {|x| - |y|} \right| \le |x|^{ - 2} } \right\}, \hfill \\
  \Omega_3  &:=& \left\{ {y \in \mathbb{R}^2 :|x - y| <|x|/2,\left| {|x| - |y|} \right| > |x|^{ - 2} } \right\}. 
\eqq
Fix $\eps>0$ small. For $|x|$ large enough,  
\bq \label{eq:neutral-TF-U1}
  \int\limits_{\Omega_1} {\frac{{\rho (y)}}
{{|x - y|}}dy}  &=& \int\limits_{|y| <  |x|^{1/2}} {\frac{{\rho (y)}}
{{|x - y|}}dy}  + \int\limits_{|x - y| \ge |x|/2,|y| \ge |x|^{1/2}} {\frac{{\rho (y)}}
{{|x - y|}}dy}  \nn\hfill \\
   &\le& {\frac{{\int \rho}}
{{|x|-|x|^{1/2}}} + } \int\limits_{|y| \ge |x|^{1/2}} {\frac{{\rho (y)}}
{{(|x|/2)}}dy}  \le \frac{\lambda +2\eps }
{{|x|}}.
\eq
Moreover, since $\rho(y)\le (2\pi |y|)^{-1}$,
\bq \label{eq:neutral-TF-U2}
\int\limits_{\Omega_2} {\frac{{\rho (y)}}
{{|x - y|}}dy}  &\le& \int\limits_{|x - y| \le \eps  } {\frac{1}
{{2\pi |y||x - y|}}dy}+  \int\limits_{|x-y|\ge \eps, \left| {|x| - |y|} \right| \le |x|^{ - 2}} {\frac{{1}}
{{2\pi |y||x - y|}}dy} \nn \hfill\\
&\le&\int\limits_{|x - y| \le \eps } {\frac{1}
{{2\pi (|x| -\eps).|x - y|}}dy}+ \int\limits_{\left| {|x| - |y|} \right| \le |x|^{ - 2} } {\frac{1}
{{(2\pi |y|)\eps}}dy} \nn \hfill\\
&=& \frac{{\eps }}
{{|x| -\eps }}+\frac{2}{\eps |x|^2} \le \frac{2\eps }
{{|x|}}.
\eq
Next, using the polar integral in 
$$ \Omega_3\subset \{y\in \R^2: 3|x|/2\ge |y| \ge |x|/2, \left| {|x| - |y|} \right| > |x|^{ - 2} \}$$
we have, with notation $s:={\min\{|x|,r\}} / {\max\{|x|,r\}}$,
\bq \label{eq:U4-polar-identity}
\int\limits_{\Omega_3} {\frac{{\rho (y)}}
{{|x - y|}}dy} \le \int\limits_{3|x|/2 \ge r \ge |x|/2,|r - |x|| \ge |x|^{ - 2} }^{} {\int\limits_0^{2\pi } {\frac{{\rho (r)r}}
{{\max \{ |x|,r\} \sqrt {1 + s^2  - 2s\cos (\theta )} }}d\theta dr}}.
\eq
The singularity of the integral w.r.t. $\theta$ (at $s\to 1^-$) is controlled by the following technical lemma (we shall prove later). 

\begin{lemma}[Upper bound on elliptic integral]\label{le:elliptic-integral} There exists a finite constant $C>0$ such that 
\[
\int\limits_0^{2\pi } {\frac{{d\theta }}
{{\sqrt {1 + s^2  - 2s\cos (\theta )} }}}  \le C(1+|\ln(1 - s)|)~~\text{for all}~0<s<1.
\]
\end{lemma}

Note that if $|r-|x||\ge |x|^{-2}$ and  $|x|\ge 1$ then   
$$1-s\ge 1-\frac{|x|-|x|^{-2}}{|x|+|x|^{-2}}=\frac{2|x|^{-3}}{1+|x|^{-3}}\ge |x|^{-3}.$$
Using \eqref{eq:U4-polar-identity} and Lemma~\ref{le:elliptic-integral} we get, for $|x|$ large enough,
\bq \label{eq:neutral-TF-U4}
  \int\limits_{\Omega_3} {\frac{{\rho (y)}}
{{|x - y|}}dy}  
 \le C_1\frac{{\ln (|x|)}}
{{|x|}}\int\limits_{3|x|/2 \ge |y| \ge |x|/2} {\rho (y)dy}
\eq
for some universal constant $C_1$. 

Putting \eqref{eq:neutral-TF-U1}, \eqref{eq:neutral-TF-U2} and \eqref{eq:neutral-TF-U4} together, we conclude that for any $\eps>0$ there exists $R=R(\eps,\rho)$ such that for any $|x|\ge R$, 
$$
(\rho *|x|^{ - 1} )(x) \le \frac{{\lambda+ 4\eps }}
{{|x|}}+C_1\frac{{\ln (|x|)}}
{{|x|}}\int\limits_{3|x|/2 \ge |y| \ge |x|/2} {\rho (y)dy}.
$$ 
\end{proof}

For completeness we provide the proof of the upper bound on the elliptic integral.
\begin{proof}[Proof of Lemma~\ref{le:elliptic-integral}] 
We just need to consider the singularity when $s\to 1^-$. Write 
\[
1 + s^2  - 2s\cos (2\theta ) = (1 + s)^2  - 2s(1 + \cos (\theta )) = (1 + s)^2  - 4s\cos ^2 (\theta /2).
\]
Denoting $k^2=4s/(1+s)^2$ and making a change of variable ($\theta \mapsto \pi-2\theta$), we need to show that
\bqq 
K(k):=\int\limits_0^{\pi /2} {\frac{{d\theta }}
{{\sqrt {1 - k^2 \sin ^2 (\theta )} }}}  = \int\limits_0^1 {\frac{{dt}}
{{\sqrt {(1 - t^2 )(1 - k^2 t^2 )} }}}  \le C|\ln (1 - k)|
\eqq
when $k\to 1^-$. This upper bound follows from the identity
\[
\int\limits_0^1 {\frac{{dt}}
{{\sqrt {(1 - t)(1 - kt)} }}}  = \frac{1}{\sqrt{k}} \ln \left( {\frac{{1 + \sqrt k }}
{{1 - \sqrt k }}} \right).
\]
\end{proof}

\begin{remark} 
The function $K(k)$ is the complete elliptic integral of the first kind. Its asymptotic behavior at $k\to 1^-$ is well known. It is (see \cite{AS70}, eq. (17.3.26), p. 591)
\[
K(k) = \frac{1}
{2}|\ln (1 - k)| + \frac{3}
{2}\ln (2) + o(1)_{k\to 1^-}.
\]
\end{remark}

\begin{proof}[Proof of Lemma~\ref{le:V-Vg2}] Recall that we are working on the region $2\Lambda \ge |x| \ge r \gg s$. We start with the triangle inequality 
\bq \label{eq:V-Vg2-bound-0}
|V - V*g_s^2 | \le \left| {\clmb^{ - 1}  - \clmb^{ - 1} *g_s^2 } \right| + \left| {\rho _\lambda ^{\TF } *\clmb^{ - 1}  - \rho _\lambda ^{\TF } *\clmb^{ - 1} *g_s^2 } \right|.
\eq
(If $V(x)=|x|^{-1}-1$ then the term invloved $\rho _\lambda ^{\TF } $ disappears.)

When $|x| \ge r \gg s \ge |y|$ using 
$$\left| {|x|^{ - 1}  - |x - y|^{ - 1} } \right| \le Cs|x|^{ - 2} $$
one has
\bq \label{eq:V-Vg2-bound-1}
\left| {\clmb^{ - 1}  - \clmb^{ - 1} *g_s^2 } \right|(x) \le \int {\left| {|x|^{ - 1}  - |x - y|^{ - 1} } \right|g_s^2 (y)dy}  \le Cs|x|^{ - 2} .
\eq
Moreover,
\bq\label{eq:V-Vg2-bound-1b}
\left| {\rho _\lambda ^{\TF } *\clmb^{ - 1}  - \rho _\lambda ^{\TF } *\clmb^{ - 1} *g_s^2 } \right|(x) \le \iint {\rho _\lambda ^{\TF } (x - y)\left| {|y|^{ - 1}  - |y - z|^{ - 1} } \right|g_s^2 (z)dydz}.
\eq
We divide the integral into two domains. If $|y|\ge r/2$ then using 
$$\left| {|y|^{ - 1}  - |y - z|^{ - 1} } \right| \le Cs|y|^{ - 2}  \le Csr^{ - 1} |y|^{ - 1} $$
and $({\rho _\lambda ^{\TF } *\clmb^{ - 1} })(x)\le C(|x|^{-1/2}+1)$ (see Lemma~\ref{le:convolution-bound}) we obtain
\bq \label{eq:V-Vg2-bound-2}
\iint\limits_{|y| \ge r/2} {\rho _\lambda ^{\TF } (x - y)\left| {|y|^{ - 1}  - |y - z|^{ - 1} } \right|g_s^2 (z)dydz} \le Csr^{ - 1} (\rho _\lambda ^{\TF } *\clmb^{ - 1} )(x) \le Csr^{ - 1} |x|^{ - 1/2} .
\eq
If $|y|\le r/2$ then using 
$$\rho _\lambda ^{\TF } (x - y) \le C(|x - y|^{ - 1}  + 1) \le C(|x|^{ - 1}  + 1)$$
and $\int_{|y|\le 2r}|y|^{-1}dy \le Cr$ we obtain
\bq \label{eq:V-Vg2-bound-3}
	&~&  \iint\limits_{|y| \le \delta /2} {\rho _\lambda ^{\TF } (x - y)\left| {|y|^{ - 1}  - |y - z|^{ - 1} } \right|g_s^2 (z)dydz} \nn \hfill \\
	&\le& C(|x|^{ - 1}  + 1)\iint\limits_{|y| \le r/2,|y - z| \le 2r} {\left( {|y|^{ - 1}  + |y - z|^{ - 1} } \right)g_s^2 (z)dydz} \le Cr (|x|^{ - 1}  + 1)
\eq
Replacing \eqref{eq:V-Vg2-bound-2} and \eqref{eq:V-Vg2-bound-3} into \eqref{eq:V-Vg2-bound-1b} we arrive at
\[
	\left| {\rho _\lambda ^{\TF } *\clmb^{ - 1}  - \rho _\lambda ^{\TF } *\clmb^{ - 1} *g_s^2 } \right|(x) \le C(sr^{ - 1} |x|^{ - 1/2}  + r|x|^{ - 1}  + 
	r)~~\text{when}~|x|\ge r.
\]
From the latter inequality and \eqref{eq:V-Vg2-bound-1} we can deduce from \eqref{eq:V-Vg2-bound-0} that 
\[
|V*g_s^2  - V|(x) \le C(s|x|^{ - 2}  + sr^{ - 1} |x|^{ - 1/2}  + r|x|^{ - 1}  + r)~~\text{when}~|x|\ge r.
\]
Taking the square integral of the previous inequality over $\{r\le |x|\le 2\Lambda\}$ we get (with $\Lambda=|\ln h|$, $r=h^{1/2}$, $s=h^{2/3}$)
\[
\int_{r \le |x| \le 2\Lambda } {|V - V*g_s^2 |^2 (x)dx}  \le C(s^2r^{ -2}\Lambda +r|\ln (\Lambda/r)| + r^2 \Lambda^2 ) \le Ch^{1/4}.
\]
\end{proof}

\end{appendix}

\paragraph{Acknowledgments:} P.T.N. and F.P. would like to thank KTH Stockholm and the University of Copenhagen, respectively, for hospitality. F.P. is grateful to A. Laptev for discussion. We thank the referee for constructive suggestions. This work was partially supported by the Danish council for independent research.



(P.T. Nam)
Department of Mathematical Sciences, University of Copenhagen, Universitetsparken 5, DK-2100 Copenhagen, Denmark. {E-mail:} ptnam@math.ku.dk 

\vspace{10pt}

(F. Portmann)
Department of Mathematics, Royal Institute of Technology, Lindstedtsv\"agen 25, SE-10044 Stockholm, Sweden. {E-mail:} fabianpo@kth.se

\vspace{10pt}

(J.P. Solovej)
Department of Mathematical Sciences, University of Copenhagen, Universitetsparken 5, DK-2100 Copenhagen, Denmark. {E-mail:} solovej@math.ku.dk

\end{document}